\documentclass[12pt,technote, onecolumn, draft]{IEEEtran}

\usepackage{xcolor}
\definecolor{TUMGray}{gray}{0.6}
\definecolor{TUMGreen}{RGB}{162,173,0}
\definecolor{TUMOrange}{RGB}{227,114,34}
\newtheorem{definition}{Definition}
\newtheorem{corollary}{Corollary}
\newtheorem{theorem}{Theorem}
\newtheorem{lemma}{Lemma}
\newtheorem{remark}{Remark}
\newtheorem{example}{Example}

\newtheorem{construction}{Construction}

\usepackage{cite}
\usepackage{amssymb}
\usepackage{amsmath}
\usepackage{tikz}
\usepackage{mathrsfs}
\usetikzlibrary{arrows, arrows.meta, shapes, shapes.multipart, positioning, calc, backgrounds, fit}

\usepackage[justification=centering]{caption} 
\usepackage[capitalize]{cleveref}

\newcommand{\N}{\mathcal{N}_{h,r,s}}
\newcommand{\GNS}{(\epsilon,\ell)\text{-}\mathcal{N}_{h,r,s}}
\newcommand{\GN}{(\epsilon,\ell)\text{-}\mathcal{N}_{h,r,\alpha\ell+\epsilon}}

\newcommand{\qs}{q^{\mathrm{MDS}}(\mathcal{N})}

\title{On the Minimum Field Size of Network MDS Codes for Generalized Combination Networks$^{\dag}$}

\author{Qin Zhou, Fang-Wei Fu
\IEEEcompsocitemizethanks{\IEEEcompsocthanksitem Qin Zhou and Fang-Wei Fu are with the Chern Institute of Mathematics and LPMC, Nankai University, Tianjin 300071, China, Emails: qinzhou@mail.nankai.edu.cn, fwfu@nankai.edu.cn.}
\thanks{$^\dag$This research is supported by the National Key Research and Development Program of China (Grant No. 2022YFA1005000), the National Natural Science Foundation of China (Grant No.  62371259), the Fundamental Research Funds for the Central Universities of China (Nankai University), and the Nankai Zhide Foundation.}
}

\begin{document}

\maketitle

\begin{abstract}
This paper investigates the minimum field size required for network maximum distance separable (MDS) codes, a critical parameter affecting computational complexity at network nodes. Focusing on generalized combination networks $(\epsilon,\ell)$-$\mathcal{N}_{h,r,\alpha\ell+\epsilon}$ and Zosin--Khuller networks $Z_{h,m,N}$, we develop a systematic framework for both scalar and vector network MDS codes. For scalar codes on $(\epsilon,\ell)$-$\mathcal{N}_{h,r,\alpha\ell+\epsilon}$, we establish an equivalence between the minimum distance of network codes and the minimum Hamming distance of classical linear codes, converting network-level MDS constraints into coding-theoretic conditions. This yields necessary and sufficient existence conditions linked to classical MDS codes and covering Grassmannian codes. Using refined greedy constructions and MRD-based designs, we obtain improved bounds on the minimum field size, outperforming the state-of-the-art universal bound. For vector network codes over $\mathbb{F}_{q}^{t}$, we develop an analogous distance equivalence and characterize existence via covering Grassmannian codes, yielding bounds on the minimum effective field size $q^{t}$. Notably, the gap between optimal scalar and vector MDS codes vanishes for several parameter regimes; we explicitly identify a family of such networks where vector coding offers no field-size advantage over scalar coding. For $Z_{h,m,N}$, we derive lower bounds on the minimal effective field size of vector MDS codes using hypergraph homomorphisms and subset-intersection arguments, strictly improving prior scalar bounds. We further provide a hypergraph-homomorphism-based necessary and sufficient condition for vector MDS construction, yielding an upper bound on the minimum field size. Finally, we bound the MDS gap between optimal scalar and vector solutions.
\end{abstract}


\section{Introduction}
Since Yeung and Zhang~\cite{YZ1999} investigated the general source-coding system with multiple sources, multiple encoders, and multiple decoders in 1999, and Ahlswede et al.~\cite{AC2000} formally introduced the concept of network coding in 2000, network coding has fundamentally overcome the limitations of traditional routing, which merely forwards data. By allowing intermediate nodes to linearly combine received information, network coding significantly improves the information transmission rate and has become a research hotspot in information theory and network communications over the past two decades. Li et al.~\cite{LY2003} confirmed via the vector space approach that linear network coding over a finite field already achieves the maximum information rate for multicast networks. Subsequently, Koetter and M\'edard~\cite{KM2003} established an algebraic framework for linear network coding, laying a solid foundation for subsequent theoretical advancements.

However, network communications are often subject to link errors caused by random noise, traffic congestion, or malicious attacks. Notably, an error occurring on a single link may propagate through downstream links, ultimately corrupting the information received at sink nodes. To address this issue, Cai and Yeung~\cite{CY2002} integrated network coding with error correction techniques, proposing the concept of network error-correcting codes and mitigating errors by introducing redundancy in the spatial domain. Building on this work, subsequent studies~\cite{YC2006,Y2006,YY2011} extended three fundamental bounds from classical coding theory, namely the Hamming bound, the Gilbert–Varshamov bound, and the Singleton bound, to the network context. Among these extensions, linear network codes that achieve the Singleton-type bound are defined as network maximum distance separable (MDS) codes. These codes exhibit optimal error-correction capability, and their construction methods have been developed in~\cite{YY2011,GF2013,GY2021,M2011,Z2008}.

The minimum field size required for network MDS codes is a critical parameter that directly affects their practical applicability. Specifically, the field size impacts the computational complexity and storage overhead at network nodes, and an excessively large field size reduces node computational efficiency and increases hardware costs, thereby hindering practical deployment of network coding. While existing construction schemes for network MDS codes applicable to general single-source multicast networks can achieve error correction, they typically require a large field size (e.g.,~\cite{GY2021}), making it challenging to meet the demands of low-complexity and efficient deployment. Therefore, deriving tight upper and lower bounds on the minimum field size of network MDS codes for specific network topologies, and designing code constructions with low field-size requirements, have become central problems in the current research of network coding.

In this paper, we investigate the minimum field size required for network MDS codes over generalized combination networks $(\epsilon,\ell)$-$\mathcal{N}_{h,r,\alpha\ell+\epsilon}$ and Zosin--Khuller (ZK) networks $Z_{h,m,N}$, two canonical layered multicast topologies for benchmarking scalar and vector network codes. Our analysis covers both scalar and vector linear network codes across distinct parameter regimes. We establish rigorous connections between network codes and classical linear codes, as well as covering Grassmannian codes. Based on these connections, we develop a unified framework that yields tight bounds on the minimum field size for both coding scenarios. The main contributions of this work are summarized as follows.

\begin{enumerate}
    \item For generalized combination networks $(\epsilon,\ell)$-$\mathcal{N}_{h,r,\alpha\ell+\epsilon}$:
    \begin{itemize}
        \item We establish rigorous equivalences between the minimum distance of scalar network codes and the Hamming distance of classical linear codes (Lemma \ref{Lem: relation between}), and between vector network codes and $\mathbb{F}_{q}$-linear codes over $\mathbb{F}_{q}^{t}$ (Lemma \ref{lem:relation between vector}). 
        \item Depending on the parameter regimes $h\leq \alpha$, $\alpha\leq h\leq \alpha\ell$, and $\alpha\ell \leq h \leq \alpha\ell+\epsilon$, we derive necessary and sufficient conditions for the existence of scalar/vector network MDS codes, linking them to classical MDS codes and covering Grassmannian codes (\cref{cor1,thm4,thm5,thm6,cor5,thm9,thm:vector_big,thm:thm vector}).
        \item Using refined greedy constructions and MRD-based designs, we obtain improved lower and upper bounds on the minimum field size (\cref{lem:upper bound of covering,lem:2lower bound of covering,lem:3lower bound of covering,lem:4lower bound,lem:greedy}, \cref{cor4,cor:cor5,cor:vector lower bound,cor:vector 1st case,cor8,cor:cor9}). Our upper bound significantly outperforms the state-of-the-art universal bound in \cite{GY2021}.
        \item We rigorously identify a family of generalized combination networks $(1,\ell)$-$\mathcal{N}_{2\ell,r,2\ell+1}$ with zero MDS gap, demonstrating that vector coding offers no field-size advantage over scalar coding in these configurations (\cref{thm:thm10}).
    \end{itemize}
    \item For Zosin–Khuller networks $Z_{h,m,N}$:
    \begin{itemize}
        \item We introduce a hypergraph homomorphism approach to characterize vector network MDS codes, yielding a necessary and sufficient condition for their existence (\cref{thm15}, \cref{thm:thm14}).
        \item For $h=2$, using hypergraph homomorphism and chromatic number arguments, we derive a lower bound on the minimum field size of vector MDS codes, which in the scalar case yields \(q^{\mathrm{MDS}}(Z_{2,m,N})\geq \psi\left(\frac{\binom{N}{m}}{\left\lfloor\frac{N}{m}\right\rfloor}-1\right)\), where \(\psi(x)\) denotes the smallest prime power that is greater than or equal to \(x\), strictly improving the prior bound from \cite{WS2024} (\cref{thm17}). 
        \item For parameters satisfying $h<N/(N-m)$, applying subset-intersection arguments, we obtain a lower bound $q_{v}^{\mathrm{MDS}}(Z_{h,m,N})\geq \psi\left(\binom{N}{m}-h+1\right)$, which also improves the bound in \cite{WS2024} (\cref{thm:thm13}).
        \item For general $h$, we provide upper bounds on the MDS gap between optimal scalar and vector solutions (\cref{cor:gap}).
        \end{itemize}
\end{enumerate}
Collectively, these results complete the theoretical characterization of feasible MDS coding for generalized combination networks and clarify the inherent performance limits of vector network coding in reducing alphabet size.

The remainder of this paper is organized as follows. Section II provides the basic notation and definitions used throughout this paper. Section III derives bounds for scalar network MDS codes in the network $(\epsilon,\ell)$-$\mathcal{N}_{h,r,\alpha\ell+\epsilon}$. Section IV extends the scalar results to the vector setting and presents a subclass of networks where the MDS gap is explicitly shown to be zero. Section V investigates vector network MDS codes in Zosin–Khuller networks and proposes an upper bound on the gap for this network. Finally, Section VI concludes the paper with a summary and discussion of open problems.

\section{Preliminaries}
Let $\mathcal{G}=(\mathcal{V}, \mathcal{E})$ denote a finite directed acyclic graph with vertex set $\mathcal{V}$ and edge set $\mathcal{E}$, where multiple parallel edges between any two vertices are allowed. An edge \(e=(i, j) \in \mathcal{E}\) represents a link from node $i$ to node $j$, and we use $\operatorname{tail}(e)$ and $\operatorname{head}(e)$ to denote its tail node $i$ and head node $j$, respectively. Additionally, \(\operatorname{In}(i)\) denotes the set of incoming edges of node $i$, and $\operatorname{Out}(i)$ denotes the set of outgoing edges of node $i$.  

In this paper, we consider single-source multicast networks over $\mathcal{G}$, which can be represented by a triple $\mathcal{N}=(\mathcal{G},\sigma,R)$. Here, $\sigma \in \mathcal{V}$ is the single source node, and $R \subset \mathcal{V}\setminus\{\sigma\}$ is a non-empty set of sink nodes. The source node $\sigma$ broadcasts $h$ independent source messages $x_1, x_2, \cdots, x_h$ to all sink nodes in $R$. Without loss of generality, assume that the source node has no incoming edges, each sink node has no outgoing edges, and every link has unit capacity. To facilitate subsequent analysis, we introduce $h$ imaginary source edges connected to $\sigma$, denoted by $d_1, d_2, \cdots, d_h$. This allows us to define $\operatorname{In}(\sigma)$ as $\{d_1, d_2, \cdots, d_h\}$, where the $i$-th source message $x_i$ is transmitted to $\sigma$ via the imaginary edge $d_i$.

In such a network, the coding operation proceeds as follows. For each edge $e \in \mathcal{E}$, the node $\operatorname{tail}(e)$ first receives all packets from its incoming edges $\operatorname{In}(\operatorname{tail}(e))$; the edge $e$ then transmits a function of these received packets. Network coding is called linear if all such functions are linear. This paper focuses on linear network coding. Let $\mathbb{F}_q$ denote the finite field of order $q$. We further distinguish two kinds of coding schemes: scalar network coding and vector network coding. In the former case, both the source messages and the packets carried by edges are elements of $\mathbb{F}_q$; in the latter case, both the messages and the packets are vectors of length $t$ over $\mathbb{F}_q$, where $t$ denotes a positive integer.

\subsection{Scalar Network Error-Correction Coding}

For a scalar linear network code $\mathcal{C}$, the packet $u_e$ carried by $e$ can be calculated as
\[
u_e = \sum_{d \in \operatorname{In}(\operatorname{tail}(e))} u_d k_{d,e},
\]
where $k_{d,e} \in \mathbb{F}_q$, and $(k_{d,e}: d \in \operatorname{In}(\operatorname{tail}(e)))^T$ is called the local encoding vector of the edge $e$. Since $u_e$ is a linear combination of the messages $x_1, x_2, \cdots, x_h$, there exists a column vector $f_e \in \mathbb{F}_q^h$ such that
\[
u_e = (x_1, x_2, \cdots, x_h) \cdot f_e,
\]
where $f_e$ is called the global encoding vector of $e$ and is determined by the local encoding vectors. For each sink node $\gamma \in R$, denote
\[
F(\gamma) \triangleq (f_e)_{e \in \operatorname{In}(\gamma)} \in \mathbb{F}_q^{h \times |\operatorname{In}(\gamma)|},
\]
and call it the decoding matrix at $\gamma$. Then $\gamma$ receives the packets $(u_e: e \in \operatorname{In}(\gamma)) = (x_1, x_2, \cdots, x_h) \cdot F(\gamma)$ and can decode all messages if and only if $F(\gamma)$ has full row rank, i.e., $\operatorname{rank}(F(\gamma)) = h$. We say that $\mathcal{C}$ is decodable if for each sink node $\gamma \in R$, the matrix $F(\gamma)$ has full row rank.

If an error occurs on edge $e$, the head node $\operatorname{head}(e)$ receives a corrupted packet $\tilde{u}_e = u_e + z_e$, where $u_e \in \mathbb{F}_q$ is the packet that is supposed to be transmitted over $e$ and $z_e \in \mathbb{F}_q$ is the additive error. We treat $z_e$ as the error message and call the vector $\mathbf{z}=(z_e)_{e\in\mathcal{E}}$ the error message vector. An error pattern $\rho \subseteq \mathcal{E}$ is a set of links in which errors occur. An error message vector $\mathbf{z}$ matches an error pattern $\rho$ if $z_e=0$ for all $e\notin \rho$, namely, $\operatorname{supp}(\mathbf{z})\subseteq \rho$, which is denoted by $\mathbf{z}\in\rho$. The corrupted packet $\tilde{u}_e$ can be calculated as
\[
\tilde{u}_e = \sum_{d \in \operatorname{In}(\operatorname{tail}(e))} \tilde{u}_d k_{d,e} + z_e.
\]
Similarly, $\tilde{u}_e$ can also be expressed as a linear combination of both source messages and error messages. Let $\mathbf{x}=(x_1,x_2,\cdots,x_h)$ be the source message vector and $\mathbf{z}=(z_e:e\in\mathcal{E})$ be the error message vector. Then
\[
\tilde{u}_e = (\mathbf{x}, \mathbf{z}) \cdot \tilde{f}_e,
\]
where the vector $\tilde{f}_e \in \mathbb{F}_q^{h+|\mathcal{E}|}$ is called the extended global encoding vector of $e$. For each sink node $\gamma \in R$, denote
\[
\tilde{F}(\gamma) \triangleq (\tilde{f}_e)_{e \in \operatorname{In}(\gamma)} \in \mathbb{F}_q^{(h+|\mathcal{E}|) \times |\operatorname{In}(\gamma)|},
\]
and call it the extended global encoding matrix at $\gamma$. This matrix admits a natural block decomposition, separating source-related and error-related components, i.e.,
\[
\tilde{F}(\gamma) = \begin{pmatrix} F(\gamma) \\ G(\gamma) \end{pmatrix},
\]
where $F(\gamma) \in \mathbb{F}_q^{h \times |\operatorname{In}(\gamma)|}$ is the original decoding matrix of $\gamma$, and $G(\gamma) \in \mathbb{F}_q^{|\mathcal{E}| \times |\operatorname{In}(\gamma)|}$ is the error-related submatrix.

Assume $\mathcal{C}$ is decodable; otherwise, even with no errors, at least one sink node would fail to decode the source messages. For any two vectors $y_{\gamma}, y_{\gamma}^{\prime} \in \mathbb{F}_{q}^{|\operatorname{In}(\gamma)|}$, the distance between them is defined as
\[
d^{\gamma}(y_{\gamma}, y_{\gamma}^{\prime}) = \min \{ |\rho| : \exists \text{ an error vector } \mathbf{z} \text{ such that } y_{\gamma} - y_{\gamma}^{\prime} = \mathbf{z} \cdot G(\gamma) \text{ and } \operatorname{supp}(\mathbf{z}) \subseteq \rho \}.
\]
One can verify that the pair $(\mathbb{F}_{q}^{|\operatorname{In}(\gamma)|}, d^{(\gamma)}(\cdot, \cdot))$ forms a metric space. Building on this metric, the minimum distance of the code $\mathcal{C}$ at the sink node $\gamma$ is defined as
\begin{equation}\label{eq:definition}
d_{\min}(\mathcal{C}, \gamma) \triangleq \min_{\substack{x, x' \in \mathbb{F}_{q}^{h} \\ x \neq x'}} d^{(\gamma)}(x \cdot F(\gamma), x' \cdot F(\gamma)).
\end{equation}
This minimum distance directly characterizes the error-correcting capability of $\mathcal{C}$ at $\gamma$; that is, the sink node $\gamma$ can correct up to $\lfloor (d_{\min}(\mathcal{C}, \gamma) - 1) / 2 \rfloor$ link errors.

For two nodes $u$ and $v$, a cut separating $v$ from $u$ is a set of edges whose removal leaves no path from $u$ to $v$. The capacity of a cut is the number of edges in the set, and the minimum cut capacity separating $v$ from $u$ is the minimum such value over all cuts. Let $C_{\gamma}$ be the minimum cut capacity separating $\gamma$ from $\sigma$. In \cite[Theorem 2]{GF2013}, a Singleton-type bound on the minimum distance $d_{\min}(\mathcal{C}, \gamma)$ is presented:
\begin{equation}\label{eq:Singleton}
d_{\min}(\mathcal{C}, \gamma) \leq C_{\gamma} - h + 1.
\end{equation}
A decodable code $\mathcal{C}$ is called a network maximum distance separable (MDS) code if it achieves the bound in \eqref{eq:Singleton} with equality for all sink nodes. Such codes are optimal in error correction. For a given network $\mathcal{N}$, let $\qs$ denote the minimum field size required to construct a scalar network MDS code.

\subsection{Vector Network Error-Correction Coding}

For a vector linear network code $\mathcal{C}$, the packet $\tilde{u}_e \in \mathbb{F}_q^t$ transmitted over edge $e$ can be expressed as
\[
\tilde{u}_e = \sum_{d \in \operatorname{In}(\operatorname{tail}(e))} \tilde{u}_d k_{d,e} + z_e,
\]
where $k_{d,e} \in \mathbb{F}_q^{t \times t}$ is the local encoding coefficient, and $z_e \in \mathbb{F}_q^t$ is the error on edge $e$. Let $\mathbf{x} \in \mathbb{F}_q^{th}$ be the source message vector, and let $\mathbf{z} \in \mathbb{F}_q^{t|\mathcal{E}|}$ be the error message vector. Then the corrupted packet $\tilde{u}_e$ can be rewritten as a linear combination of the source message vector $\mathbf{x}$ and the error message vector $\mathbf{z}$, i.e.,
\[
\tilde{u}_e = (\mathbf{x}, \mathbf{z}) \cdot \tilde{f}_e,
\]
where $\tilde{f}_e \in \mathbb{F}_q^{t(h+|\mathcal{E}|) \times t}$ is the extended global encoding vector of $e$. For each sink node $\gamma \in R$, the extended global encoding matrix of $\gamma$ has the form:
\[
\tilde{F}(\gamma) \triangleq (\tilde{f}_e)_{e \in \operatorname{In}(\gamma)} = \begin{pmatrix} F(\gamma) \\ G(\gamma) \end{pmatrix},
\]
where $F(\gamma) \in \mathbb{F}_q^{th \times t|\operatorname{In}(\gamma)|}$ is the decoding matrix of $\gamma$, and $G(\gamma) \in \mathbb{F}_q^{t|\mathcal{E}| \times t|\operatorname{In}(\gamma)|}$ is the error-related matrix.

For any two vectors $y_{\gamma}, y_{\gamma}^{\prime} \in \mathbb{F}_{q}^{t|\operatorname{In}(\gamma)|}$, define their distance as
\[
d^{\gamma}(y_{\gamma}, y_{\gamma}^{\prime}) = \min \{ |\rho| : \exists \text{ an error vector } \mathbf{z} \text{ such that } y_{\gamma} - y_{\gamma}^{\prime} = \mathbf{z} \cdot G(\gamma) \text{ and } \operatorname{supp}(\mathbf{z}) \subseteq \rho \}.
\]
The minimum distance of a decodable vector code $\mathcal{C}$ at the sink node $\gamma$ is defined as
\begin{equation}
d_{\min}(\mathcal{C}, \gamma) \triangleq \min_{\substack{x, x' \in \mathbb{F}_{q}^{th} \\ x \neq x'}} d^{(\gamma)}(x \cdot F(\gamma), x' \cdot F(\gamma)).
\end{equation}
This minimum distance $d_{\min}(\mathcal{C}, \gamma)$ is similarly constrained by a Singleton-type bound (see \cite[Lemma 15]{WS2024}):
\begin{align}\label{eq: v Singleton}
d_{\min}(\mathcal{C}, \gamma) \leq C_{\gamma} - h + 1.
\end{align}
A decodable vector code $\mathcal{C}$ is called a vector network MDS code if $d_{\min}(\mathcal{C}, \gamma) = C_{\gamma} - h + 1$ for every $\gamma \in R$. For a given network $\mathcal{N}$, let $q_{v}^{\mathrm{MDS}}(\mathcal{N})$ denote the minimum effective field size for vector network MDS codes, defined as the minimum value of $q^{t}$ such that a vector network MDS code over $\mathbb{F}_{q}^{t}$ exists. Note that a scalar network code over $\mathbb{F}_{q}$ can be equivalently viewed as a vector network code over $\mathbb{F}_{q}$ with vector length $t = 1$. This implies $q_{v}^{\mathrm{MDS}}(\mathcal{N}) \leq q^{\mathrm{MDS}}(\mathcal{N})$. Furthermore, vector network coding offers greater flexibility in choosing encoding coefficients, so $q_{v}^{\mathrm{MDS}}(\mathcal{N})$ may be strictly smaller than $q^{\mathrm{MDS}}(\mathcal{N})$. To quantify this advantage, define the MDS gap of the network $\mathcal{N}$ as
\[
\mathrm{gap}^{\mathrm{MDS}}(\mathcal{N}) \triangleq q^{\mathrm{MDS}}(\mathcal{N}) - q_{v}^{\mathrm{MDS}}(\mathcal{N}).
\]

\section{Scalar Error-Correcting Codes for Generalized Combination Networks}

In this section, we focus on a generalization of the combination network $\N$, denoted by $\GNS$. This network has played a pivotal role in demonstrating the advantages of vector coding over scalar coding, as shown in \cite{EW2018,LW2021}. As a special case, the combination network $\N$ has been extensively studied \cite{CC2020,ML2012,XM2007}, including research on error correction \cite{WS2024}. Here we extend this line of research to investigate error correction in the generalized combination network $\GNS$.

The generalized combination network $\GNS$, which is shown in \cref{fig:GNmodel}, is a network with three layers. In the first layer, there is a single source node $\sigma$, which multicasts $h$ independent source messages to all sink nodes. In the middle layer, there are $r$ nodes, which are indexed by the elements of $[r] \triangleq \{1, 2, \dots, r\}$. The source node $\sigma$ connects to each middle-layer node via $\ell$ parallel links. For any $\alpha$-subset of the middle-layer nodes, where $\alpha = (s - \epsilon) / \ell$, there is exactly one sink node connected to these $\alpha$ middle-layer nodes via $\ell$ parallel links. In total, there are $\binom{r}{\alpha}$ sink nodes, each indexed by a unique $\alpha$-subset of $[r]$ and requiring all $h$ source messages. The set of all sink nodes is denoted by $\binom{[r]}{\alpha}$. Additionally, the source node $\sigma$ directly connects to each sink node via $\epsilon$ parallel links. For each sink node, the total number of incoming links is $s = \alpha\ell + \epsilon$. The total number of links in the network $\GNS$ is $L = r\ell + \binom{r}{\alpha}(\alpha\ell + \epsilon)$. In what follows, we will characterize the generalized combination network $\GNS$ using the parameter $\alpha$ instead of $s$, since $\alpha$ more naturally captures the network's combinatorial structure.

\begin{figure}[h]
    \centering
\def\x{0.55}
\begin{tikzpicture}
  [font=\normalsize,>=stealth',
  mycircle/.style={circle, draw=TUMGray, very thick, text width=.1em, minimum height=1.5em, text centered},
  mylink/.style={color=black, thick, ->, bend right=5},
  mylink_r/.style={color=black, thick, ->, bend left=5},
  myarc/.style={color=TUMOrange, thick},
  myarc1/.style={color=TUMGreen, thick},
  myarc2/.style={color=blue, thick},
 ]
  \coordinate (Source node) at (0*\x,4*\x);
  {\node[mycircle,label=above:{$\sigma$}] (Source node) {};}
  \node[draw=none,above right = \x*3pt and \x*5pt of Source node] {$\mathbf{x}_{1},\mathbf{x}_{2},\cdots,\mathbf{x}_{h}$};
  \node[mycircle,below left = \x*40pt and \x*100pt of Source node] (M0) {};
  \node[mycircle,right = \x*20pt of M0] (M1) {};
  \node[mycircle,right = \x*20pt of M1] (M2) {};
  \node[draw=none,right = \x*20pt of M2] (M3) {$\cdots$};
  \node[mycircle,right = \x*20pt of M3] (M4) {};
  \node[mycircle,right = \x*20pt of M4]
  (M5) {};
  \node[draw=none,right= 15pt of M5]  {$r$ middle layer nodes};
   \path[] (Source node) edge[mylink] (M0.north)
  edge[mylink]  (M1.north)
  edge[mylink] node [midway] (link1)
  {} (M2.north)
  edge[mylink]  (M4.north)
  edge[mylink] (M5.north);
  \path[] (Source node) edge[mylink_r] (M0.north)
  edge[mylink_r] (M1.north)
  edge[mylink_r] (M2.north)
  edge[mylink_r] (M4.north)
  edge[mylink_r] (M5.north);
  \node[draw=none,right= 5pt of link1]  {$\cdots$};
  \node[mycircle,below left = \x*60pt and \x*20pt of M0] (S0) {};
  \node[draw=none,right = \x*100pt of S0] (S1) {$\cdots$};
  \node[mycircle,right= 100pt of S1, label=right:{$\binom{r}{\alpha}$ sink nodes}] (S2) {};
   \path[] (M0) edge[mylink] (S0.north);
   \path[] (M0) edge[mylink_r] (S0.north);
   \path[] (M1) edge[mylink] (S0.70);
   \path[] (M1) edge[mylink_r] node [pos=0.3] (Tlink1) {} (S0.70);
   \path[] (M3) edge[mylink] (S0.45);
   \path[] (M3) edge[mylink_r] (S0.45);
   \node[draw=none,right = \x*3pt of Tlink1]  {$\cdots$};
  
   \path[] (M5) edge[mylink] node [pos=0.3] (Tlink2) {} (S2.north);
   \path[] (M5) edge[mylink_r] (S2.north);
   \path[] (M3) edge[mylink] (S2.110);
   \path[] (M3) edge[mylink_r] (S2.110);
   \node[draw=none,left = \x*0pt of Tlink2] {$\cdots$};
   \draw[mylink, bend right=50] (Source node) to (S0.100);
   \draw[mylink, bend right=60] (Source node) to (S0.100);
   \draw[mylink, bend left=50] (Source node) to (S2.80);
   \draw[mylink, bend left=60] (Source node) to (S2.80);
    \draw [myarc] ($(Source node)-(\x*90pt,\x*25pt)$) arc (160:230:\x*15pt);
    \node[draw=none, myarc] at ($(M0.45)+(\x*16pt, \x*4pt)$) {$\ell$};
    \draw [myarc] ($(Source node)-(\x*50pt,\x*25pt)$) arc (190:250:\x*15pt);
    \draw [myarc] ($(Source node)-(\x*18pt,\x*35pt)$) arc (240:290:\x*15pt);
    \draw [myarc] ($(Source node)+(\x*60pt,\x*-35pt)$) arc (270:330:\x*15pt);
    \draw [myarc] ($(Source node)+(\x*100pt,\x*-40pt)$) arc (310:370:\x*15pt);
    
    \draw [myarc] ($(M0)-(\x*30pt,\x*30pt)$) arc (220:280:\x*15pt);
    \draw [myarc] ($(M1)-(\x*45pt,\x*25pt)$) arc (190:250:\x*15pt);
    \draw [myarc] ($(M3)-(\x*90pt,\x*25pt)$) arc (160:230:\x*15pt);
    \draw [myarc] ($(M3)+(\x*55pt,\x*-35pt)$) arc (270:340:\x*15pt);
    \draw [myarc] ($(M5)+(\x*13pt,\x*-45pt)$) arc (290:350:\x*15pt);

    \draw [myarc1] ($(Source node)-(\x*90pt,\x*-5pt)$) arc (170:230:\x*15pt);
    \node[draw=none, myarc1] at ($(Source node)-(\x*100pt, \x*-10pt)$) {$\epsilon$};
    \draw [myarc1] ($(Source node)+(\x*100pt,\x*-5pt)$) arc (280:340:\x*15pt);
    
    \draw [myarc2] ($(Source node)-(\x*150pt,\x*95pt)$) arc (210:270:\x*50pt);
    \node[draw=none, myarc2] at ($(Source node)-(\x*100pt, \x*125pt)$) {$\alpha\ell$};

    \draw [myarc2] ($(Source node)+(\x*130pt,\x*-115pt)$) arc (270:330:\x*50pt);
    \node[draw=none, myarc2] at ($(Source node)+(\x*120pt, \x*-125pt)$) {$\alpha\ell$};
\end{tikzpicture}
\caption{A depiction of the $\GN$ network.}
   \label{fig:GNmodel}
 \end{figure}
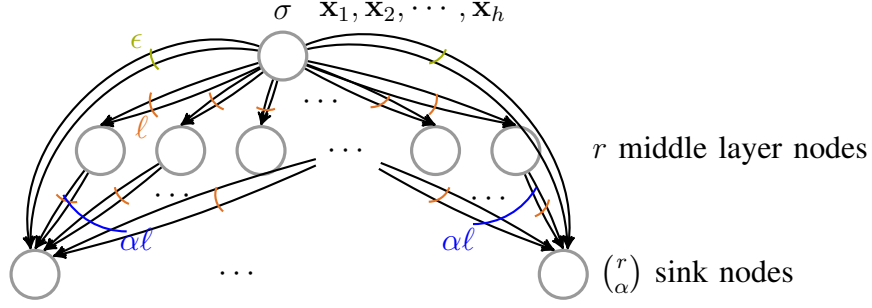 

When considering error-correcting codes over the network $\GN$, we may assume without loss of generality that each middle node simply forwards the received packets to sink nodes. We now prove this claim. Let $t$ be a positive integer. Consider an arbitrary vector network MDS code $\mathcal{C}$ over $\mathbb{F}_q^t$ on the network $\GN$. For each middle node $i \in [r]$ and each of the $\ell$ parallel links connecting $\sigma$ to this middle node, let $A_{ij} \in \mathbb{F}_q^{ht \times t}$ denote the local encoding vector for the $j$-th link. For each $i \in [r]$, concatenate these matrices to form
\[
A_i = \begin{pmatrix} A_{i1} & A_{i2} & \cdots & A_{i\ell} \end{pmatrix} \in \mathbb{F}_q^{ht \times \ell t}.
\]
Consider a sink node $\gamma = \{i_1, i_2, \dots, i_\alpha\} \subseteq [r]$. For each $k \in [\alpha]$ and each $j \in [\ell]$, let $B_{i_k}^{j} \in \mathbb{F}_q^{\ell t \times t}$ denote the local encoding vector for the $j$-th link which connects the middle node $i_k$ and $\gamma$. For each $k \in [\alpha]$, concatenate these matrices to form
\[
B_{i_k} = \begin{pmatrix} B_{i_k}^{1} & B_{i_k}^{2} & \cdots & B_{i_k}^{\ell} \end{pmatrix} \in \mathbb{F}_q^{\ell t \times \ell t}.
\]
For each $j \in [\epsilon]$, let $C_j \in \mathbb{F}_q^{ht \times t}$ denote the local encoding vector for the $j$-th link which connects the source node $\sigma$ and $\gamma$. Concatenate these matrices to form
\[
C = \begin{pmatrix} C_1 & C_2 & \cdots & C_{\epsilon} \end{pmatrix} \in \mathbb{F}_q^{ht \times \epsilon t}.
\]
To simplify the extended global encoding matrix of $\gamma$, define
\(
A = \begin{pmatrix} A_{i_1} & A_{i_2} & \cdots & A_{i_\alpha} \end{pmatrix}
\)
and the block-diagonal matrix
\[
\operatorname{diag}(B) =
\begin{pmatrix}
B_{i_1} & \mathbf{0} & \cdots & \mathbf{0} \\
\mathbf{0} & B_{i_2} & \cdots & \mathbf{0} \\
\vdots & \vdots & \ddots & \vdots \\
\mathbf{0} & \mathbf{0} & \cdots & B_{i_\alpha}
\end{pmatrix}.
\]
Then at the sink node $\gamma$, the extended global encoding matrix of the code $\mathcal{C}$ is
\[
\begin{pmatrix} F(\gamma) \\ G(\gamma) \end{pmatrix}
=
\begin{pmatrix}
A \cdot \operatorname{diag}(B) & C \\
\operatorname{diag}(B) & \mathbf{0} \\
I_{\alpha \ell t} & \mathbf{0} \\
\mathbf{0} & I_{\epsilon t} \\
\mathbf{0} & \mathbf{0}
\end{pmatrix},
\]
where $I_n \in \mathbb{F}_q^{n \times n}$ is the identity matrix, and $\mathbf{0}$ denotes the all-zero matrix of appropriate dimensions.

We now show that if we modify $\mathcal{C}$ to a new code $\mathcal{C}^{\prime}$ where all middle nodes only route packets, then $\mathcal{C}^{\prime}$ remains a vector network MDS code over $\mathbb{F}_q^t$. We prove this via contradiction. Suppose $\mathcal{C}^{\prime}$ is not a vector network MDS code. Then there exists at least one sink node $\gamma = \{i_1, i_2, \dots, i_\alpha\}$ such that $\mathcal{C}^{\prime}$ does not achieve the bound in \eqref{eq: v Singleton} at $\gamma$. The extended global encoding matrix of $\mathcal{C}^{\prime}$ at $\gamma$ is
\[
\begin{pmatrix} F^{\prime}(\gamma) \\ G^{\prime}(\gamma) \end{pmatrix}
=
\begin{pmatrix}
A & C \\
I_{\alpha \ell t} & \mathbf{0} \\
I_{\alpha \ell t} & \mathbf{0} \\
\mathbf{0} & I_{\epsilon t} \\
\mathbf{0} & \mathbf{0}
\end{pmatrix}.
\]
Since each row of $G^{\prime}(\gamma)$ has at most one nonzero entry, there exist two distinct source message vectors $\mathbf{x}^1, \mathbf{x}^2 \in \mathbb{F}_q^{ht}$ and an error vector $\mathbf{z} \in (\mathbb{F}_q^t)^{\alpha\ell+\epsilon}$ with Hamming weight $\mathrm{wt}_H(\mathbf{z}) < \alpha\ell + \epsilon - h + 1$ such that
\[
\mathbf{x}^1 \cdot \begin{pmatrix} A & C \end{pmatrix} - \mathbf{x}^2 \cdot \begin{pmatrix} A & C \end{pmatrix} = \mathbf{z} \cdot I_{(\alpha\ell+\epsilon)t}.
\]
Partition $\mathbf{z}$ to match the block structure of $\begin{pmatrix} A & C \end{pmatrix}$: let $\mathbf{z} = \begin{pmatrix} \mathbf{z}_1 & \mathbf{z}_2 \end{pmatrix}$, where $\mathbf{z}_1 \in (\mathbb{F}_q^t)^{\alpha\ell}$ and $\mathbf{z}_2 \in (\mathbb{F}_q^t)^{\epsilon}$. Then we have
\begin{align*}
& \mathbf{x}^1 \cdot \begin{pmatrix} A \cdot \operatorname{diag}(B) & C \end{pmatrix} - \mathbf{x}^2 \cdot \begin{pmatrix} A \cdot \operatorname{diag}(B) & C \end{pmatrix} \\
&= \begin{pmatrix} \mathbf{z}_1 \cdot \operatorname{diag}(B) & \mathbf{z}_2  \end{pmatrix} \\
&= \mathbf{z} \cdot \begin{pmatrix} \operatorname{diag}(B) & \mathbf{0} \\ \mathbf{0} & I_{\epsilon t} \end{pmatrix}.
\end{align*}
This implies that in the original code $\mathcal{C}$, the distance between the received packets corresponding to $\mathbf{x}^1$ and $\mathbf{x}^2$ at the sink node $\gamma$ is less than $\alpha\ell + \epsilon - h + 1$, contradicting the fact that $\mathcal{C}$ is a vector network MDS code. Thus, any vector network MDS code can be transformed into one in which all middle nodes simply forward packets, over the same finite field.

\subsection{Bounds on the Size of Covering Grassmannian Codes}\label{sec:covering_bounds}

To establish the forthcoming bounds on the minimum field size for network MDS codes, we first recall the concept of covering Grassmannian codes and present several lemmas that provide upper and lower bounds on the size of such codes under specific parameters. These bounds will be used in the subsequent subsections to derive existence conditions and field-size bounds for network MDS codes.

\begin{definition}[Covering Grassmannian Codes \cite{EZ2019}] 
Let $\mathcal{G}(n, k)$ denote the set of all $k$-dimensional subspaces of $\mathbb{F}_q^n$. An $\alpha$-$(n, k, \delta)_{q}^{c}$ covering Grassmannian code $C$ is a subset of $\mathcal{G}(n, k)$ such that each $\alpha$-element subset of $C$ spans a subspace whose dimension is at least $\delta+k$ in $\mathbb{F}_q^n$. Additionally, let $B_q(n, k, \delta; \alpha)$ denote the maximum possible size of an $\alpha$-$(n, k, \delta)_{q}^{c}$ covering Grassmannian code.
\end{definition}

\begin{lemma}\label{lem:upper bound of covering}
Let $n$, $k$, and $\alpha$ be positive integers such that $\alpha \geq 2$ and $n \geq \alpha k$. Then the following recursive bound holds:
\[
B_q(n,k,(\alpha-1)k;\alpha) \leq B_q(n-k,k,(\alpha-2)k;\alpha-1) + 1,
\]
and consequently,
\[
B_q(n,k,(\alpha-1)k;\alpha) \leq \left\lfloor \frac{q^{n-(\alpha-2)k}-1}{q^k-1} \right\rfloor + \alpha - 2.
\]
\end{lemma}

\begin{IEEEproof}
For any $i \ge 2$, let $C$ be an $(i+1)$-$(n, k, ik)_q^c$ covering Grassmannian code achieving the maximum cardinality, i.e., $|C| = B_q(n, k, ik; i+1)$. Fix an arbitrary codeword $c_1 \in C$. There exists an $(n-k)$-dimensional subspace $W \subseteq \mathbb{F}_q^n$ such that $\mathbb{F}_q^n = c_1 \oplus W$ and $c_1 \cap W = \{\mathbf{0}\}$. Let $\varphi: W \to \mathbb{F}_q^{n-k}$ be a linear isomorphism. For each $c_j \in C \setminus \{c_1\}$, define
\[
c_j' = \left\{ \varphi(\vec{v}_2) : \vec{v}_1 + \vec{v}_2 \in c_j \text{ such that } \vec{v}_1 \in c_1 \text{ and } \vec{v}_2 \in W \right\}, \quad 2 \le j \le B_q(n, k, ik; i+1).
\]
It follows that $\dim(c_j') = k$ and $c_j' \subseteq \mathbb{F}_q^{n-k}$. One can verify that the collection $\{c_j' : 2 \le j \le B_q(n, k, ik; i+1)\}$ forms an $i$-$(n-k, k, (i-1)k)_q^c$ covering Grassmannian code. Therefore, we obtain
\[
B_q(n, k, ik; i+1) \le B_q(n-k, k, (i-1)k; i) + 1, \quad \forall \ i \ge 2.
\]
Iterating this inequality yields
\begin{align}\label{eq: size of covering}
B_q(n, k, (\alpha-1)k; \alpha) \leq B_q(n-(\alpha-2)k, k, k; 2) + (\alpha-2).
\end{align}
From \cite{QW2023}, the following upper bound for $B_q(n-(\alpha-2)k, k, k; 2)$ is known:
\[
B_q(n-(\alpha-2)k, k, k; 2) \leq \left\lfloor \frac{q^{n-(\alpha-2)k} - 1}{q^k - 1} \right\rfloor.
\]
Substituting this into \eqref{eq: size of covering}, we conclude that
\[
B_q(n, k, (\alpha-1)k; \alpha) \leq \left\lfloor \frac{q^{n-(\alpha-2)k} - 1}{q^k - 1} \right\rfloor + \alpha-2.
\]
\end{IEEEproof}

In the special case where \(\alpha = 2\), the covering Grassmannian code reduces to the classical subspace code setting. The following result, adapted from \cite{EV2011}, gives a lower bound on the maximum size of such codes.

\begin{lemma}{\cite[Theorem 11]{EV2011}}\label{lem:2lower bound of covering}
Let $n$, $k$ be positive integers such that $n\geq 2k$. Then the maximum size of a $2$-$(n,k,k)_{q}^{c}$ covering Grassmannian code satisfies:
\[
B_{q}(n, k, k; 2)\geq q^{k} \lfloor \frac{q^{n-k}-1}{q^k-1} \rfloor +1.
\]
\end{lemma}

For two matrices $A, B \in \mathbb{F}_q^{m \times k}$, the rank distance between them is defined as $d(A,B) := \operatorname{rank}(A - B)$. A linear subspace $\mathcal{C} \subseteq \mathbb{F}_q^{m \times k}$ is called a linear rank-metric code with parameters $[m \times k, K, d]$ if it has dimension $K$ over $\mathbb{F}_q$ and satisfies $d(C_1, C_2) \ge d$ for all distinct $C_1, C_2 \in \mathcal{C}$. Any such code satisfies the following Singleton-type bound in rank-metric coding theory:
\begin{equation}\label{eq:rank_singleton}
K \le \min\left\{m(k - d + 1),\ k(m - d + 1)\right\}.
\end{equation}
Codes achieving this bound with equality are known as maximum rank distance (MRD) codes, and are known to exist for all admissible parameters \cite{D1978}.

Let $f(x) = f_0 + f_1x + \dots + f_{k-1}x^{k-1} + x^k \in \mathbb{F}_q[x]$ be a primitive polynomial of degree $k$ over $\mathbb{F}_q$. Its companion matrix is given by
\[
D = \begin{pmatrix}
0 & 0 & \dots & 0 & -f_0 \\
1 & 0 & \dots & 0 & -f_1 \\
0 & 1 & \dots & 0 & -f_2 \\
\vdots & \vdots & \ddots & \vdots & \vdots \\
0 & 0 & \dots & 1 & -f_{k-1}
\end{pmatrix}.
\]
Define the set of matrices generated by polynomials in \(D\) as
\[
\mathscr{C} = \{a_0 I_k + a_1 D + \dots + a_{k-1} D^{k-1} \mid a_0, a_1, \dots, a_{k-1} \in \mathbb{F}_q\},
\]
where $I_k$ denotes the $k \times k$ identity matrix. Since $f(x)$ is primitive, $\mathscr{C}$ forms a matrix field isomorphic to $\mathbb{F}_{q^k}$. Every nonzero matrix in $\mathscr{C}$ is invertible; hence, any two distinct matrices in $\mathscr{C}$ have rank distance $k$, and the minimum rank distance of $\mathscr{C}$ is $k$. Note that $|\mathscr{C}| = q^k$ achieves the Singleton bound in \eqref{eq:rank_singleton}. Consequently, $\mathscr{C}$ is an MRD code with parameters $[k \times k, k, k]$. This constructed MRD code will serve as a key building block for establishing the lower bounds in \cref{lem:3lower bound of covering} and \cref{lem:4lower bound} below, whose proofs are deferred to Appendix A.

\begin{lemma}\label{lem:3lower bound of covering}
Let $n$, $k$ be positive integers such that $n \geq 3k$ and $k \mid n$. Then the maximum size of a $3$-$(n,k,2k)_{q}^{c}$ covering Grassmannian code satisfies
\[
B_q(n,k,2k;3) \geq \frac{q^{(\lfloor (n/k-3)/2 \rfloor + 2)k} - 1}{q^k - 1} + 2\left( \frac{n/k-3}{2} - \left\lfloor \frac{n/k-3}{2} \right\rfloor \right).
\]
\end{lemma}

\begin{remark}
For the case where $k \nmid n$, the bound in \cref{lem:3lower bound of covering} can be adapted by replacing $n/k$ with $\lfloor n/k \rfloor$.
\end{remark}

\begin{lemma}\label{lem:4lower bound}
Let $n$, $k$, $\alpha$ be positive integers such that $\alpha \geq 4$ and $n \geq \alpha k$. Then the maximum size of an $\alpha$-$(n,k,(\alpha-1)k)_{q}^{c}$ covering Grassmannian code satisfies
\[
B_q(n,k,(\alpha-1)k;\alpha) \geq q^k + \left\lfloor \frac{n}{k} \right\rfloor - \alpha + 1.
\]
\end{lemma}

In addition to the MRD-based constructions above, we provide another lower bound via a refined greedy construction. The proof is deferred to Appendix B.

\begin{lemma}\label{lem:greedy}
Let $q$ be a prime power, and let $n,k,\alpha$ be positive integers satisfying $\alpha \ge 3$ and $n \ge \alpha k$. For any non-negative integer $m$, define
\[
\begin{aligned}
T_1 &= q^{n}-q^{(\alpha-1)k}, \\
T_2 &= (m-\alpha)(q^k-1)+\sum_{j=\alpha-1}^{m-2}\binom{j}{\alpha-2}(q^k-1)\big(q^{(\alpha-2)k}-1\big), \\
T_3 &= (q^{k-1}-1)+\binom{m-1}{\alpha-1}(q^{k-1}-1)\big(q^{(\alpha-1)k}-1\big).
\end{aligned}
\]
Then there exists a maximum non-negative integer $m$ such that
\[
T_1 - T_2 - T_3 > 0.
\]
Correspondingly, there exists an $\alpha$-$(n,k,(\alpha-1)k)_q^c$ covering Grassmannian code in $\mathcal{G}_q(n,k)$ with exactly $m$ codewords.
\end{lemma}

\begin{remark}
We now compare our lower bounds on $B_q(n,k,(\alpha-1)k;\alpha)$ with the known bound in \cite{QW2023}. The known bound, derived from \cite[Theorem V.1]{QW2023}, is based on a greedy selection restricted to subspaces within a fixed spread and requires maximizing $m$ such that
\[
\frac{q^n-1}{q^k-1} - \sum_{i=0}^{m-\alpha} \binom{\alpha-2+i}{\alpha-2} \frac{q^{(\alpha-1)k}-1}{q-1} \ge 1.
\]
Our refined greedy construction for $\alpha\ge3$ removes this restriction by considering all $k$-dimensional subspaces in $\mathcal{G}(n,k)$, and generally yields a tighter bound. For $\alpha=3$, we provide an MRD-based construction in Lemma \ref{lem:3lower bound of covering}; for $\alpha\ge4$, a similar construction is given in Lemma \ref{lem:4lower bound}. These two bounds and the refined greedy bound are complementary, with their relative tightness depending on the parameters $(q,n,k,\alpha)$. Numerical comparisons show that both of our bounds improve upon the known bound in \cite{QW2023}. A detailed numerical comparison is given in Table~\ref{tab:bound_compare}.

\begin{table}[htbp]
\centering
\caption{Comparison of lower bounds on $B_q(n,k,(\alpha-1)k;\alpha)$}
\label{tab:bound_compare}
\begin{tabular}{cccccccc}
\hline
$q$ & $n$ & $k$ & $\alpha$ & Known bound \cite{QW2023} &Lemma~\ref{lem:3lower bound of covering} &Lemma~\ref{lem:4lower bound}& Lemma~\ref{lem:greedy}  \\
\hline
2 & 8 & 2 & 3 & 4 &6 & -& 5  \\
2 & 10 & 2 & 3 & 8 & 21& -&10  \\
2 & 20 & 5 & 3 & 9 & 34&- & 12  \\

3 & 18 & 3 & 4 & 22 & -& 30&25  \\
3 & 36 & 6 & 4 & 187 &-&732 & 214  \\
4 & 24 &4 &5& 14 & -& 258 & 15 \\
5& 23 & 3 & 6 &20 &-&127&21\\
5& 30 & 3 &6 & 166 &-&130 & 174\\

\hline
\end{tabular}
\end{table}

\end{remark}

\subsection{Scalar Error-Correcting Codes in the $\GN$ Network}

We now consider scalar error-correcting codes for the network $\GN$. We first define the local encoding vectors of a scalar network code $\mathcal{C}$, then characterize the extended global encoding matrix at each sink node, and finally establish a critical connection between the network code's minimum distance and the Hamming distance of a classical linear code.

For each $i \in [r]$ and $j \in [\ell]$, let $v_{i,j} \in \mathbb{F}_q^h$ be the local encoding vector on the $j$-th incoming edge of the middle node $i$. Further, for each middle node $i \in [r]$, denote the source-to-middle encoding matrix
\[
v_i = \begin{pmatrix} v_{i,1} & v_{i,2} & \cdots & v_{i,\ell} \end{pmatrix}.
\]
For a sink node $\gamma = \{i_1, i_2, \dots, i_\alpha\} \subseteq [r]$, let $a_\gamma \in \mathbb{F}_q^{h \times \epsilon}$ denote the direct-link encoding matrix: its columns are the local encoding vectors of the $\epsilon$ parallel direct links from $\sigma$ to $\gamma$. Then the extended global encoding matrix at $\gamma$ is
\begin{equation}\label{eq:G}
\begin{pmatrix} F(\gamma) \\ G(\gamma) \end{pmatrix}
=
\begin{pmatrix}
V & a_\gamma \\
I_{\alpha \ell} & \mathbf{0} \\
I_{\alpha \ell} & \mathbf{0}\\
\mathbf{0} & I_{\epsilon} \\
\mathbf{0} & \mathbf{0}
\end{pmatrix},
\end{equation}
where \(
V = \begin{pmatrix} v_{i_1} & v_{i_2} & \cdots & v_{i_\alpha} \end{pmatrix}
\).

Let \(\mathcal{C}(\gamma)\) be the classical linear code over \(\mathbb{F}_q\) generated by the matrix \(F(\gamma)\). The following lemma establishes that the minimum distance of the scalar network code \(\mathcal{C}\) at \(\gamma\) is exactly the minimum Hamming distance of \(\mathcal{C}(\gamma)\).

\begin{lemma}\label{Lem: relation between}
Let \(\mathcal{C}\) be a decodable scalar linear network code for the generalized combination network \(\GN\). Then for each sink node \(\gamma \in \binom{[r]}{\alpha}\),
\[
d_{\min}(\mathcal{C}, \gamma) = d_H(\mathcal{C}(\gamma)).
\]
\end{lemma}

\begin{IEEEproof}
From \eqref{eq:G}, the rows of \(G(\gamma)\) corresponding to the edges in \(\operatorname{In}(\gamma)\) form the identity matrix \(I_{\alpha\ell+\epsilon}\). For any two distinct source message vectors \(x, x' \in \mathbb{F}_q^h\), let \(\Delta = x \cdot F(\gamma) - x' \cdot F(\gamma)\). Since \(I_{\alpha\ell+\epsilon}\) is a submatrix of \(G(\gamma)\), there exists an error vector \(\mathbf{z}\) such that \(\Delta = \mathbf{z} \cdot G(\gamma)\) and \(|\operatorname{supp}(\mathbf{z})| = wt_H(\Delta)\). Thus,
\[
\begin{aligned}
d^{\gamma}(x \cdot F(\gamma), x' \cdot F(\gamma))
&= \min \{ |\rho| : \exists \text{ an error vector } \mathbf{z} \text{ such that } \Delta = \mathbf{z} \cdot G(\gamma) \text{ and } \operatorname{supp}(\mathbf{z}) \subseteq \rho \} \\
&\leq wt_H(x \cdot F(\gamma) - x' \cdot F(\gamma)) \\
&= d_H(x \cdot F(\gamma), x' \cdot F(\gamma)).
\end{aligned}
\]
From \eqref{eq:definition}, we obtain
\[
\begin{aligned}
d_{\min}(\mathcal{C}, \gamma)
&= \min_{\substack{x, x' \in \mathbb{F}_q^h \\ x \neq x'}} d^{\gamma}(x \cdot F(\gamma), x' \cdot F(\gamma)) \\
&\leq \min_{\substack{x, x' \in \mathbb{F}_q^h \\ x \neq x'}} d_H(x \cdot F(\gamma), x' \cdot F(\gamma)) \\
&= d_H(\mathcal{C}(\gamma)).
\end{aligned}
\]

On the other hand, for any distinct \(x, x' \in \mathbb{F}_q^h\), let \(\rho(x, x')\) denote the error pattern that minimizes \(|\rho|\) in the definition of \(d^{\gamma}(x \cdot F(\gamma), x' \cdot F(\gamma))\). From \eqref{eq:G}, each row of \(G(\gamma)\) has at most one nonzero entry. This implies that \(|\rho(x, x')| \ge wt_H(x \cdot F(\gamma) - x' \cdot F(\gamma))\). Thus,
\[
\begin{aligned}
d_{\min}(\mathcal{C}, \gamma)
&= \min_{\substack{x, x' \in \mathbb{F}_q^h \\ x \neq x'}} d^{\gamma}(x \cdot F(\gamma), x' \cdot F(\gamma)) \\
&= \min_{\substack{x, x' \in \mathbb{F}_q^h \\ x \neq x'}} |\rho(x, x')| \\
&\ge \min_{\substack{x, x' \in \mathbb{F}_q^h \\ x \neq x'}} wt_H(x \cdot F(\gamma) - x' \cdot F(\gamma)) \\
&= d_H(\mathcal{C}(\gamma)).
\end{aligned}
\]
Combining the two inequalities above yields
\[
d_{\min}(\mathcal{C}, \gamma) = d_H(\mathcal{C}(\gamma)).
\]
\end{IEEEproof}

Further, let \(\mathcal{C}(r)\) be a linear code over \(\mathbb{F}_q\) generated by the matrix \(F \triangleq \begin{pmatrix} v_1 & v_2 & \cdots & v_r \end{pmatrix}\), where \(v_i\) is the source-to-middle encoding matrix for the middle node \(i\). The following lemma connects the minimum Hamming distance of the code \(\mathcal{C}(r)\) to the minimum distance of the scalar network code \(\mathcal{C}\).

\begin{lemma}\label{thm3}
Let \(\mathcal{C}\) be a scalar linear network code for the generalized combination network \(\GN\), where \(h \le \alpha\). Let \(d\) be a positive integer satisfying
\(
\alpha\ell + \epsilon - \alpha + 1 \le d \le \alpha\ell + \epsilon - h + 1
\). If \(\mathcal{C}\) is decodable and its minimum distance satisfies \(d_{\min}(\mathcal{C}, \gamma) \ge d\) for each sink node \(\gamma \in \binom{[r]}{\alpha}\), then the code \(\mathcal{C}(r)\) is an \([r\ell, h, \ge (r-\alpha)\ell - \epsilon + d]_q\) linear code.
\end{lemma}

\begin{IEEEproof}
By \cref{Lem: relation between}, the condition \(d_{\min}(\mathcal{C}, \gamma) \ge d\) implies \(d_H(\mathcal{C}(\gamma)) \ge d\) for each \(\gamma \in \binom{[r]}{\alpha}\). Thus, any \(\alpha\ell+\epsilon-d+1\) columns of \(F(\gamma)\) form a matrix of full row rank. Otherwise, there exists a nonzero codeword in \(\mathcal{C}(\gamma)\) with Hamming weight at most \(d-1\), contradicting \(d_H(\mathcal{C}(\gamma)) \ge d\). Observe that \(F(\gamma)\) is a submatrix of \(F\). Since \(F(\gamma)\) has full row rank, \(F\) also has full row rank, equal to \(h\).

Since $\alpha\ell+\epsilon-d+1 \le \alpha$, any $\alpha\ell+\epsilon-d+1$ columns of $F$ must lie within the columns indexed by some $\alpha$-subset of middle nodes, and hence are contained in some $F(\gamma)$. Therefore, any \(\alpha\ell+\epsilon-d+1\) columns of \(F\) form a matrix of full row rank. This implies that every nonzero codeword in \(\mathcal{C}(r)\) has Hamming weight at least
\[
r\ell - (\alpha\ell+\epsilon-d+1) + 1 = (r-\alpha)\ell - \epsilon + d.
\]
Thus, the proof is complete.
\end{IEEEproof}

Combining Lemmas \ref{Lem: relation between} and \ref{thm3} yields the following necessary condition for the existence of scalar network MDS codes.

\begin{theorem}\label{cor1}
Let \(\mathcal{C}\) be a scalar linear network code for the generalized combination network \(\GN\), where \(h \le \alpha\). If \(\mathcal{C}\) is MDS, then \(\mathcal{C}(r)\) is MDS, and for each \(\gamma \in \binom{[r]}{\alpha}\), \(\mathcal{C}(\gamma)\) is MDS.
\end{theorem}

\begin{corollary}\label{cor:lower_bounds}
For the generalized combination network \(\mathcal{N}=\GN\) with \(2 \le h \le \alpha\), the minimum field size \(\qs\) satisfies
\[
\qs \ge \psi(r\ell - h + 1) \quad \text{and} \quad \qs \ge \psi(\alpha\ell + \epsilon - h + 1),
\]
where \(\psi(x)\) denotes the smallest prime power that is greater than or equal to \(x\).
\end{corollary}

\begin{IEEEproof}
By Theorem~\ref{cor1}, if a linear network MDS code over \(\mathbb{F}_q\) for \(\GN\) exists, then there must exist two linear MDS codes over \(\mathbb{F}_q\) of lengths \(r\ell\) and \(\alpha\ell+\epsilon\), both of dimension \(h\). Since \(h \ge 2\), applying the Griesmer bound to such linear MDS codes yields \(q \ge r\ell - h + 1\) and \(q \ge \alpha\ell + \epsilon - h + 1\) \cite[pp. 340–341]{R2006}.
\end{IEEEproof}

\begin{theorem}\label{thm4}
For the generalized combination network \(\GN\) with \(h \le \alpha\), a linear network MDS code over \(\mathbb{F}_q\) exists if and only if an \([n_{\max}, h]_q\) MDS code exists, where \(n_{\max} = \max(\alpha\ell + \epsilon, r\ell)\).
\end{theorem}

\begin{IEEEproof}
The ``only if" direction follows from \cref{cor1}. For the ``if" direction, let \(G\) be a generator matrix of an \([n_{\max}, h]_q\) MDS code with columns \(g_1, \dots, g_{n_{\max}}\). Assign \(v_{i,j} = g_{(i-1)\ell + j}\) for each middle node \(i\) and link \(j\). For each sink node \(\gamma = \{i_1, \dots, i_\alpha\}\), select any \(\epsilon\) columns from the remaining \(n_{\max} - \alpha\ell\) columns of \(G\) to form the direct-link encoding matrix \(a_\gamma\), and let middle nodes forward. The decoding matrix at each sink node consists of \(\alpha\ell+\epsilon\) columns of \(G\), hence generates an MDS code. By \cref{Lem: relation between}, the constructed network code is MDS.
\end{IEEEproof}

\begin{corollary}\label{cor2}
For the generalized combination network \(\mathcal{N}=\GN\) with \(h \le \alpha\), the minimum field size \(\qs\) satisfies
\[
\qs \le \psi(n_{\max} - 1),
\]
where \(n_{\max} = \max(\alpha\ell + \epsilon, r\ell)\).
\end{corollary}

\begin{IEEEproof}
For each prime power \(q \ge n_{\max} - 1\), an \([n_{\max}, h]_q\) MDS code exists (e.g., via extended Reed–Solomon codes). By Theorem \ref{thm4}, there exists a network MDS code over \(\mathbb{F}_q\) for \(\GN\) with \(h \le \alpha\). Hence \(\qs \le \psi(n_{\max} - 1)\).
\end{IEEEproof}

We next consider the case where the parameters satisfy $\alpha\ell \le h \le \alpha\ell+\epsilon$. Denote
\[
S(\mathcal{C}) \triangleq \{ \langle v_i \rangle \mid 1 \le i \le r \},
\]
where $\langle v_i \rangle$ denotes the subspace over $\mathbb{F}_q$ spanned by the columns of the source-to-middle encoding matrix $v_i$. The following lemma relates the network code's minimum distance to the structural properties of $S(\mathcal{C})$.

\begin{lemma}\label{lem3}
Let \(\mathcal{C}\) be a scalar linear network code for the generalized combination network \(\GN\), where \(\alpha\ell \le h \le \alpha\ell+\epsilon\). Let $d$ be a positive integer satisfying
\(
\epsilon - h + 1 \le d \le \alpha\ell + \epsilon - h + 1
\). If \(\mathcal{C}\) is decodable and its minimum distance satisfies \(d_{\min}(\mathcal{C}, \gamma) \ge d\) for each sink node \(\gamma \in \binom{[r]}{\alpha}\), then the subspace collection \(S(\mathcal{C})\) satisfies the following two properties:
\begin{enumerate}
\item Each subspace in \(S(\mathcal{C})\) has dimension at most \(\ell\);
\item Any \(\alpha\)-subset of \(S(\mathcal{C})\) spans a subspace of dimension at least \(h - \epsilon + d - 1\).
\end{enumerate}
\end{lemma}

\begin{IEEEproof}
For each $i \in [r]$, the matrix $v_i$ has $\ell$ columns, so the dimension of its column span is at most $\ell$.

For Property 2), consider an arbitrary $\alpha$-subset of $S(\mathcal{C})$, which corresponds to middle nodes $i_1, i_2, \dots, i_\alpha$. Let $\gamma = \{i_1, i_2, \dots, i_\alpha\}$ be the sink node connected to these $\alpha$ middle nodes. By \cref{Lem: relation between}, $d_{\min}(\mathcal{C}, \gamma) \ge d$ implies $d_H(\mathcal{C}(\gamma)) \ge d$. It follows that any $\alpha\ell + \epsilon - d + 1$ columns of $F(\gamma)$ form a matrix of full row rank. Since $h \ge \alpha\ell$ and $d \le \alpha\ell + \epsilon - h + 1$, we have $\alpha\ell + \epsilon - d + 1 \ge \alpha\ell$. Thus, there exists an $h \times (\alpha\ell + \epsilon - d + 1)$ submatrix $\hat{F}(\gamma)$ of $F(\gamma)$ satisfying both conditions below:
\begin{itemize}
    \item $\hat{F}(\gamma)$ contains $\begin{pmatrix} v_{i_1} & v_{i_2} & \cdots & v_{i_\alpha} \end{pmatrix}$ as a submatrix;
    \item $\hat{F}(\gamma)$ has rank $h$.
\end{itemize}
Since $d \ge \epsilon - h + 1$, we have $\epsilon - d + 1 \le h$. Therefore, the rank of $\begin{pmatrix} v_{i_1} & v_{i_2} & \cdots & v_{i_\alpha} \end{pmatrix}$ is at least $h - \epsilon + d - 1$, which implies that the dimension of the span of the $\alpha$ subspaces is at least $h - \epsilon + d - 1$.
\end{IEEEproof}

The following result follows directly from \cref{Lem: relation between} and \cref{lem3}.

\begin{theorem}\label{thm5}
Let \(\mathcal{C}\) be a linear network MDS code over \(\mathbb{F}_q\) for the generalized combination network \(\GN\), where \(\alpha\ell \le h \le \alpha\ell+\epsilon\). Then \(S(\mathcal{C})\) is an \(\alpha\)-\((h, \ell, (\alpha-1)\ell)_q^c\) covering Grassmannian code, and for each sink node \(\gamma\), \(\mathcal{C}(\gamma)\) is MDS.
\end{theorem}

\begin{corollary}\label{cor:covering_necessary}
Let $\alpha$, $\ell$, $\epsilon$, $h$, and $r$ be positive integers such that $\alpha \geq 2$ and $\alpha\ell \leq h \leq \alpha\ell+\epsilon$. If $\GN$ admits a linear network MDS code over $\mathbb{F}_q$, then
\[
r \le \left\lfloor \frac{q^{h-(\alpha-2)\ell} - 1}{q^{\ell} - 1} \right\rfloor + \alpha - 2
\quad \text{and} \quad
\alpha\ell + \epsilon - h + 1 \le q.
\]
\end{corollary}

\begin{IEEEproof}
This result follows directly from Theorem \ref{thm5} and Lemma \ref{lem:upper bound of covering}.
\end{IEEEproof}

The following theorem provides a necessary and sufficient condition for the existence of scalar network MDS codes in the parameter range \(\alpha\ell \le h \le \alpha\ell+\epsilon\).

\begin{theorem}\label{thm6}
For the generalized combination network \(\GN\) with \(\alpha\ell \le h \le \alpha\ell+\epsilon\), a linear network MDS code over \(\mathbb{F}_q\) exists if and only if both of the following hold:
\begin{enumerate}
    \item There exists an \(\alpha\)-\((h, \ell, (\alpha-1)\ell)_q^c\) covering Grassmannian code with at least \(r\) codewords;
    \item There exists an \([\alpha\ell+\epsilon, h]_q\) MDS code.
\end{enumerate}
\end{theorem}

\begin{IEEEproof}
The ``only if" direction follows immediately from \cref{thm5}.

For the ``if" direction, we construct a scalar network code explicitly using the two given components. Let \(\{s_1, s_2, \dots, s_r\}\) be a subset of the \(\alpha\)-\((h, \ell, (\alpha-1)\ell)_q^c\) covering Grassmannian code. Then for each \(i \in [r]\), \(s_i\) is an \(\ell\)-dimensional subspace of \(\mathbb{F}_q^h\). Choose a basis of \(s_i\) and form the \(h \times \ell\) matrix \(v_i\) whose columns are these basis vectors; this matrix serves as the source-to-middle encoding matrix for middle node \(i\). Let \(G\) be a generator matrix of the \([\alpha\ell+\epsilon, h]_q\) MDS code \(C\). Partition \(G\) as \(G = \begin{pmatrix} G_1 & G_2 \end{pmatrix}\), where \(G_1 \in \mathbb{F}_q^{h \times h}\) and \(G_2 \in \mathbb{F}_q^{h \times (\alpha\ell+\epsilon-h)}\). Since any \(h\) columns of \(G\) are linearly independent, \(G_1\) is invertible.

For any sink node \(\gamma = \{i_1, i_2, \dots, i_\alpha\}\), denote
\[
V_1(\gamma) = \begin{pmatrix} v_{i_1} & v_{i_2} & \cdots & v_{i_\alpha} \end{pmatrix}.
\]
By the definition of covering Grassmannian codes, \(\operatorname{rank}(V_1(\gamma)) = \alpha\ell\). Since \(h \ge \alpha\ell\), we can choose \(h - \alpha\ell\) linearly independent vectors to form \(V_2(\gamma)\) such that
\(
V(\gamma) = \begin{pmatrix} V_1(\gamma) & V_2(\gamma) \end{pmatrix}
\) is invertible. Then \(V(\gamma)G_1^{-1}G\) is an alternative generator matrix of \(C\):
\[
V(\gamma)G_1^{-1}G = \begin{pmatrix} V_1(\gamma) & V_2(\gamma) & V(\gamma)G_1^{-1}G_2 \end{pmatrix}.
\]
Select the last \(\epsilon\) columns of \(V(\gamma)G_1^{-1}G\) as the local encoding vectors on the direct links associated with \(\gamma\). By construction, the decoding matrix at each sink node is a generator matrix of the MDS code \(C\). Hence, by \cref{Lem: relation between}, the constructed network code is MDS.
\end{IEEEproof}

Combining Lemmas \ref{lem:2lower bound of covering}--\ref{lem:greedy} yields the following corollary.

\begin{corollary}\label{cor4}
Let $\alpha$, $\ell$, $\epsilon$, $h$, and $r$ be positive integers such that $\alpha\ell \le h \le \alpha\ell+\epsilon$. If \(q \ge \alpha\ell+\epsilon-1\) and either of the following conditions holds:
\begin{itemize}
    \item $\alpha=2$ and $q^{\ell} \left\lfloor \frac{q^{h-\ell}-1}{q^\ell-1} \right\rfloor + 1 \ge r$;
    \item $\alpha=3$ and $\frac{q^{(\lfloor \frac{h/\ell-3}{2} \rfloor + 2)\ell}-1}{q^{\ell}-1} + 2\left( \frac{h/\ell-3}{2} - \left\lfloor \frac{h/\ell-3}{2} \right\rfloor \right) \ge r$;
    \item $\alpha\ge 4$ and $q^{\ell} + \left\lfloor \frac{h}{\ell} \right\rfloor - \alpha + 1 \ge r$;
    \item $\alpha\ge 3$ and $T_1 - T_2 - T_3 > 0$, where $T_1,T_2,T_3$ are as defined in Lemma~\ref{lem:greedy} with $n=h$, $k=\ell$, and $m=r$.
\end{itemize}
then a linear network MDS code over \(\mathbb{F}_q\) exists for the network \(\GN\).
\end{corollary}

\begin{example}
\begin{enumerate}
    \item Consider the generalized combination network $(3,1)$-$\mathcal{N}_{4,10,6}$, which satisfies the assumptions of Corollary~\ref{cor4}. The necessary condition \(q \ge \alpha\ell+\epsilon-1\) gives \(q \ge 5\). Since $\alpha=3$, the second condition in Corollary~\ref{cor4} holds for $q\ge 8$, while the fourth condition holds for $q\ge 5$. Therefore, by Corollary~\ref{cor4}, a linear network MDS code over $\mathbb{F}_q$ exists for all $q\ge 5$, certified by the fourth condition.
    \item Consider the generalized combination network $(3,2)$-$\mathcal{N}_{8,100,9}$, which satisfies the assumptions of Corollary~\ref{cor4}. The necessary condition gives $q\ge 8$. Since $\alpha=3$, the second condition holds for $q\ge 11$, while the fourth holds for $q\ge 17$. Thus, Corollary~\ref{cor4} guarantees a linear network MDS code over $\mathbb{F}_q$ for all $q\ge 11$ via the second condition.
    \item Consider the generalized combination network $(4,1)$-$\mathcal{N}_{9,50,9}$, which satisfies the assumptions of Corollary~\ref{cor4}. The necessary condition gives $q\ge 8$. Since $\alpha=5$, the third condition holds for $q\ge 47$, while the fourth holds for $q\ge 13$. Thus, Corollary~\ref{cor4} guarantees a linear network MDS code over $\mathbb{F}_q$ for all $q\ge 13$ via the fourth condition.
    \item Consider the generalized combination network $(5,2)$-$\mathcal{N}_{11,150,13}$, which satisfies the assumptions of Corollary~\ref{cor4}. The necessary condition gives $q\ge 12$. Since $\alpha=4$, the third condition holds for $q\ge 13$, while the fourth holds for $q\ge 29$. Thus, Corollary~\ref{cor4} guarantees a linear network MDS code over $\mathbb{F}_q$ for all $q\ge 13$ via the third condition.
\end{enumerate}
\end{example}

We now analyze the properties of scalar network codes for the generalized combination network $\GN$ under the parameter constraint \(\alpha \le h \le \alpha\ell\).

\begin{lemma}\label{thm7}
Let \(\mathcal{C}\) be a scalar linear network code for the generalized combination network \(\GN\), where \(\alpha \le h \le \alpha\ell\). If \(\mathcal{C}\) is decodable and \(d_{\min}(\mathcal{C}, \gamma) \ge d\) for each sink node \(\gamma \in \binom{[r]}{\alpha}\), then
\begin{enumerate}
    \item\label{result1} When \(\epsilon+1 \le d \le \alpha\ell+\epsilon-h+1\), for any \(\alpha\) source-to-middle encoding matrices \(v_{i_1}, \dots, v_{i_\alpha}\), the matrix
    \[
    V = \begin{pmatrix} v_{i_1} & v_{i_2} & \cdots & v_{i_\alpha} \end{pmatrix}
    \]
    is the generator matrix of an \([\alpha\ell, h, \ge d-\epsilon]_q\) linear code.
    \item\label{result2} When \(\max\{\epsilon-h+1, 1\} \le d \le \epsilon+1\), the collection \(S(\mathcal{C})\) satisfies
    \begin{enumerate}
        \item each subspace in \(S(\mathcal{C})\) has dimension at most \(\ell\);
        \item any \(\alpha\) subspaces in \(S(\mathcal{C})\) span a subspace of dimension at least \(h-\epsilon+d-1\).
    \end{enumerate}
\end{enumerate}
\end{lemma}

\begin{IEEEproof}
For any sink node \(\gamma = \{i_1, i_2, \dots, i_\alpha\}\), the decoding matrix at \(\gamma\) is \(F(\gamma) = \begin{pmatrix} V & a_\gamma \end{pmatrix}\), where \(a_\gamma\) is the direct-link encoding matrix. By \cref{Lem: relation between}, \(d_{\min}(\mathcal{C}, \gamma) \ge d\) implies \(d_H(\mathcal{C}(\gamma)) \ge d\). Hence, any \(\alpha\ell+\epsilon-d+1\) columns of \(F(\gamma)\) form a matrix of full row rank. Since \(\alpha\ell+\epsilon-d+1 \le \alpha\ell\), any \(\alpha\ell+\epsilon-d+1\) columns of \(V\) also form a matrix of full row rank. Thus, the code generated by \(V\) is an \([\alpha\ell, h, \ge d-\epsilon]_q\) linear code.

Item 2) follows by adapting the subspace dimension argument used in the proof of \cref{lem3}.
\end{IEEEproof}

Combining Lemma \ref{Lem: relation between} and Lemma \ref{thm7} yields the following theorem.

\begin{theorem}\label{cor5}
Let \(\mathcal{C}\) be a scalar linear network code for the generalized combination network \(\GN\), where \(\alpha \le h \le \alpha\ell\). If \(\mathcal{C}\) is MDS, then for any sink node \(\gamma = \{i_1, i_2, \dots, i_\alpha\}\), \(\mathcal{C}(\gamma)\) is MDS, and
\(
V = \begin{pmatrix} v_{i_1} & v_{i_2} & \cdots & v_{i_\alpha} \end{pmatrix}
\) is the generator matrix of an \([\alpha\ell, h]\) MDS code.
\end{theorem}

\begin{corollary}\label{cor:cor5}
Let $\alpha$, $\ell$, $\epsilon$, $h$, and $r$ be positive integers such that \(h \ge 2\ell\) and \(\alpha \le h \le \alpha\ell\). If \(\GN\) admits a linear network MDS code over \(\mathbb{F}_q\), then
\[
r \le \left\lfloor \frac{q^{h-(\lfloor h/\ell \rfloor-2)\ell} - 1}{q^{\ell} - 1} \right\rfloor + \left\lfloor \frac{h}{\ell} \right\rfloor - 2
\quad \text{and} \quad
\alpha\ell + \epsilon - h + 1 \le q.
\]
\end{corollary}

\begin{IEEEproof}
Since \(2 \le \lfloor h/\ell \rfloor \le \alpha\), Theorem \ref{cor5} implies that \(\{\langle v_i \rangle \}_{i \in [r]}\) forms an \(\lfloor h/\ell \rfloor\)-\((h, \ell, (\lfloor h/\ell \rfloor-1)\ell)_q^c\) covering Grassmannian code. The bound on \(r\) then follows from Lemma \ref{lem:upper bound of covering}.
\end{IEEEproof}

Let \(n_{\max} = \max(\alpha\ell+\epsilon, r\ell)\), and assume that there exists an \([n_{\max}, h]_q\) MDS code. Using the same approach as in Theorem \ref{thm4}, we can construct a linear network MDS code over \(\mathbb{F}_q\) for \(\GN\) with \(\alpha \le h \le \alpha\ell\). Accordingly, we derive the following upper bound on the minimum field size:
\[
q^{\mathrm{MDS}}(\GN) \le \psi(n_{\max} - 1).
\]

\begin{remark}
A general upper bound on the minimum field size of network MDS codes was given in \cite{GY2021}: for a network \(\mathcal{N}\) with \(h\) source messages and each sink node \(\gamma \in R\) satisfying \(C_\gamma \ge h\),
\[
q^{\mathrm{MDS}}(\mathcal{N}) \le \sum_{\gamma \in R} \left| \mathscr{A}_\gamma(C_\gamma - h) \right|.
\]
As shown in \cite[Corollary 10]{GY2021}, this bound can be expressed explicitly as
\begin{equation}\label{eq:best_bound}
q^{\mathrm{MDS}}(\mathcal{N}) \le \sum_{\gamma \in R} \binom{|\operatorname{In}(\gamma)|}{C_\gamma - h}.
\end{equation}
For the network \(\GN\), the bound in \eqref{eq:best_bound} evaluates to \(\binom{r}{\alpha}\binom{\alpha\ell+\epsilon}{h}\), which is typically much larger than our bound. In contrast, our results show that an MDS code exists whenever \(q \ge \max(\alpha\ell+\epsilon, r\ell) - 1\) for \(h \le \alpha\ell\). For the parameter range \(\alpha\ell \le h \le \alpha\ell+\epsilon\), Corollary~\ref{cor4} guarantees valid MDS constructions with an even smaller required field size. Thus, our upper bound is significantly smaller than the bound in \eqref{eq:best_bound} for this network.
\end{remark}

\section{Vector Error-Correcting Codes for Generalized Combination Networks}

\subsection{\(\mathbb{F}_q\)-Linear Codes over \(\mathbb{F}_q^t\)}
We first recall essential background on $\mathbb{F}_q$-linear codes over $\mathbb{F}_q^t$, which will be extensively used in our subsequent analysis of vector error-correcting codes.

Let \(q\) be a prime power, and let \(n, t\) be positive integers. We consider \(\mathbb{F}_q^t\)-valued vectors of length \(n\). For any vector \(\mathbf{x} = (x_1, x_2, \dots, x_n) \in (\mathbb{F}_q^t)^n\), its Hamming weight is defined as the number of nonzero components \(x_i\), namely
\[
\operatorname{wt}_H(\mathbf{x}) \triangleq \left| \left\{ 1 \le i \le n : x_i \neq \mathbf{0}_t \right\} \right|,
\]
where \(\mathbf{0}_t\) denotes the all-zero vector of length \(t\). For two vectors \(\mathbf{x}, \mathbf{y} \in (\mathbb{F}_q^t)^n\), their Hamming distance equals the Hamming weight of their difference, i.e.,
\[
d_H(\mathbf{x}, \mathbf{y}) \triangleq \operatorname{wt}_H(\mathbf{x} - \mathbf{y}).
\]
An \(\mathbb{F}_q\)-linear code \(C \) over \(\mathbb{F}_q^t\) of length $n$ is defined as a vector subspace of \((\mathbb{F}_q^t)^n\) over \(\mathbb{F}_q\). Its minimum Hamming distance is given by
\[
d_H(C) \triangleq \min_{\substack{x, x' \in C \\ x \neq x'}} d_H(x, x').
\]

\begin{lemma}\label{lem: vector linear}
Let $C$ be an \(\mathbb{F}_q\)-linear code over \(\mathbb{F}_q^t\) of length $n$ with generator matrix $G$. Partition $G$ blockwise as
\[
G = \begin{pmatrix} B_1 & B_2 & \cdots & B_n \end{pmatrix},
\]
where each block \(B_i\) contains exactly $t$ columns. Then the minimum Hamming distance of \(C\) is at least $d$ if and only if any \(n-d+1\) column blocks form a full row rank submatrix.
\end{lemma}

The proof of the above lemma can be found in \cite[Lemma 3]{WS2024}. For an \(\mathbb{F}_q\)-linear code \(C \subseteq (\mathbb{F}_q^t)^n\) with \(\mathbb{F}_q\)-dimension \(k'\) and minimum Hamming distance $d$, the Singleton bound shows that
\[
d \le n + 1 - \frac{k'}{t}.
\]
Codes that achieve equality in this bound are called MDS codes. For such codes, \(k'/t\) is necessarily an integer. Set \(k = k'/t\), and call such a code an \(\mathbb{F}_q\)-linear \([n, k]\) MDS code over \(\mathbb{F}_q^t\).


\subsection{Vector Network Error-Correcting Codes}

In this section, we focus on vector network coding for the generalized combination network $\GN$. As established earlier, we may assume without loss of generality that each middle node simply forwards the received packets. For a sink node \(\gamma = \{i_1, i_2, \dots, i_\alpha\}\), the extended global encoding matrix is
\[
\begin{pmatrix} F(\gamma) \\ G(\gamma) \end{pmatrix},
\]
where
\[
F(\gamma) = \begin{pmatrix} v_{i_1} & v_{i_2} & \cdots & v_{i_\alpha} & a_\gamma \end{pmatrix} \in \mathbb{F}_q^{th \times t(\alpha\ell+\epsilon)}
\]
and \(G(\gamma) \in \mathbb{F}_q^{tL \times t(\alpha\ell+\epsilon)}\). If \(G(\gamma)\) is partitioned as \(G(\gamma) = (G_{ij})_{i \in [L], j \in [\alpha\ell+\epsilon]}\) with each \(G_{ij} \in \mathbb{F}_q^{t \times t}\), then \(G(\gamma)\) has two key properties:
\begin{enumerate}
    \item There exists a subset \(A \subseteq [L]\) with \(|A| = \alpha\ell+\epsilon\) such that the submatrix \((G_{ij})_{i \in A, j \in [\alpha\ell+\epsilon]}\) is the \(t(\alpha\ell+\epsilon) \times t(\alpha\ell+\epsilon)\) identity matrix;
    \item For each \(i \in [L]\), there is at most one nonzero \(G_{ij}\).
\end{enumerate}

Let \(\mathcal{C}(\gamma)\) denote the \(\mathbb{F}_q\)-linear code over \(\mathbb{F}_q^t\) generated by \(F(\gamma)\). The following lemma bridges the vector network code's minimum distance to the Hamming distance of \(\mathcal{C}(\gamma)\):

\begin{lemma}\label{lem:relation between vector}
Let \(\mathcal{C}\) be a vector network code for the generalized combination network \(\GN\). For every sink node \(\gamma \in \binom{[r]}{\alpha}\),
\[
d_{\min}(\mathcal{C}, \gamma) = d_H(\mathcal{C}(\gamma)).
\]
\end{lemma}

Let \(\mathcal{C}(r)\) be the \(\mathbb{F}_q\)-linear code over \(\mathbb{F}_q^t\) generated by
\(
V = \begin{pmatrix} v_1 & v_2 & \cdots & v_r \end{pmatrix}
\). Following the same reasoning as in the scalar case, we obtain the following result.

\begin{theorem}\label{thm9}
For the generalized combination network \(\GN\) with \(h \le \alpha\), a linear network MDS code over \(\mathbb{F}_q^t\) exists if and only if an \(\mathbb{F}_q\)-linear \([n_{\max}, h]\) MDS code over \(\mathbb{F}_q^t\) exists, where \(n_{\max} = \max(\alpha\ell+\epsilon, r\ell)\).
\end{theorem}

For an \(\mathbb{F}_q\)-linear \([n_{\max}, h]\) MDS code over \(\mathbb{F}_q^t\), partition its generator matrix as
\[
G = \begin{pmatrix} G_1 & G_2 & \cdots & G_{n_{\max}} \end{pmatrix},
\]
where each block \(G_i\) contains exactly \(t\) columns. Since every \(G_i\) has full column rank, its column space is a \(t\)-dimensional subspace of \(\mathbb{F}_q^{ht}\). This construction yields an \(h\)-\((ht, t, (h-1)t)_q^c\) covering Grassmannian code consisting of \(n_{\max}\) codewords. Conversely, suppose we are given an \(h\)-\((ht, t, (h-1)t)_q^c\) covering Grassmannian code with \(n_{\max}\) codewords. Each codeword corresponds to a \(t\)-dimensional subspace of \(\mathbb{F}_q^{ht}\). For each such subspace, fix a basis and stack the basis vectors columnwise to form a block matrix \(G_i\) for each \(i \in \{1, 2, \dots, n_{\max}\}\). Concatenating these blocks yields
\(
G = \begin{pmatrix} G_1 & G_2 & \cdots & G_{n_{\max}} \end{pmatrix}
\), and the code generated by \(G\) is an \(\mathbb{F}_q\)-linear \([n_{\max}, h]\) MDS code over \(\mathbb{F}_q^t\). In summary, an \(\mathbb{F}_q\)-linear \([n_{\max}, h]\) MDS code over \(\mathbb{F}_q^t\) exists if and only if there exists an \(h\)-\((ht, t, (h-1)t)_q^c\) covering Grassmannian code of size \(n_{\max}\). Combining this equivalence with Lemmas \ref{lem:upper bound of covering}--\ref{lem:greedy} yields the following consequences.

\begin{corollary}\label{cor:vector lower bound}
For the generalized combination network \(\mathcal{N}=\GN\) with \(2 \le h \le \alpha\),
\[
q_v^{\mathrm{MDS}}(\mathcal{N}) \ge \psi(n_{\max} - h + 1),
\]
where \(n_{\max} = \max(\alpha\ell + \epsilon, r\ell)\).
\end{corollary}

In particular, for the generalized combination network \(\mathcal{N}=\GN\) with \(h = 2\), Corollaries \ref{cor2} and \ref{cor:vector lower bound} yield
\[
q^{\mathrm{MDS}}(\mathcal{N}) - q_v^{\mathrm{MDS}}(\mathcal{N})
\le \psi(n_{\max} - 1) - \psi(n_{\max} - 1) = 0,
\]
i.e., the gap is zero. Thus, for such networks, vector MDS codes offer no advantage over scalar MDS codes in terms of alphabet size.

\begin{corollary}\label{cor:vector 1st case}
Let \(\alpha, \ell, \epsilon, h, r, t\) be positive integers such that \(h \le \alpha\). Set \(n_{\max} = \max(\alpha\ell+\epsilon, r\ell)\). If either of the following conditions holds:
\begin{itemize}
    \item \(q^t \ge n_{\max} - 1\);
    \item \(h \ge 3\) and \(T_1 - T_2 - T_3 > 0\), where \(T_1, T_2, T_3\) are as defined in Lemma~\ref{lem:greedy} with \(n = ht\), \(k = t\), \(\alpha = h\), and \(m = n_{\max}\);
\end{itemize}
then a linear network MDS code over \(\mathbb{F}_q^t\) exists for the network \(\GN\).
\end{corollary}

\begin{example}
\begin{enumerate}
    \item Consider the generalized combination network $(1,1)$-$\mathcal{N}_{3,4,4}$, which satisfies the hypotheses of Corollary~\ref{cor:vector 1st case} with $h=3$ and $\alpha=3$. Here, $n_{\max}=4$. The first condition holds for $q^t \ge 3$, while the second is satisfied for $q^t \ge 2$. Thus, Corollary~\ref{cor:vector 1st case} guarantees a linear network MDS code over $\mathbb{F}_2$ for this network via the second condition. This implies that the gap for this network is zero.
    \item Consider the generalized combination network $(2,1)$-$\mathcal{N}_{3,6,6}$, where $h=3$, $\alpha=4$ and $n_{\max}=6$. The first condition in Corollary~\ref{cor:vector 1st case} holds for $q^t \ge 5$, while the second holds for $q^t \ge 9$. Therefore, Corollary~\ref{cor:vector 1st case} guarantees a linear network MDS code over $\mathbb{F}_5$ for this network via the first condition.
\end{enumerate}
\end{example}

Let \(S(\mathcal{C})\) denote the set of subspaces of \(\mathbb{F}_q^{ht}\) spanned by the column vectors of \(v_i\) for \(i \in [r]\). Analogous to the scalar case, we obtain the following result under the parameter constraint \(\alpha\ell \le h \le \alpha\ell+\epsilon\).

\begin{theorem}\label{thm:vector_big}
For the generalized combination network \(\GN\) with \(\alpha\ell \le h \le \alpha\ell+\epsilon\), a linear network MDS code over \(\mathbb{F}_q^t\) exists if and only if both of the following hold:
\begin{enumerate}
    \item There exists an \(\alpha\)-\((ht, \ell t, (\alpha-1)\ell t)_q^c\) covering Grassmannian code with at least \(r\) codewords;
    \item There exists an \(\mathbb{F}_q\)-linear \([\alpha\ell+\epsilon, h]\) MDS code over \(\mathbb{F}_q^t\).
\end{enumerate}
\end{theorem}

\begin{corollary}\label{cor8}
Let \(\mathcal{N}=\GN\) satisfy \(\alpha \ge 2\) and \(\alpha\ell \le h \le \alpha\ell+\epsilon\). If a linear network MDS code over \(\mathbb{F}_q^t\) exists for \(\mathcal{N}\), then
\[
r \le \left\lfloor \frac{q^{ht-(\alpha-2)\ell t} - 1}{q^{\ell t} - 1} \right\rfloor + \alpha - 2
\quad \text{and} \quad
\alpha\ell + \epsilon - h + 1 \le q^t.
\]
\end{corollary}

\begin{corollary}\label{cor:cor9}
Let \(\mathcal{N}=\GN\) satisfy \(\alpha \ge 2\) and \(\alpha\ell \le h \le \alpha\ell+\epsilon\). If \(q^t \ge \alpha\ell+\epsilon-1\) and either of the following conditions holds:
\begin{itemize}
    \item \(\alpha = 2\) and \(q^{\ell t} \left\lfloor \frac{q^{ht-\ell t}-1}{q^{\ell t}-1} \right\rfloor + 1 \ge r\);
    \item \(\alpha = 3\) and \(\frac{q^{(\lfloor (h/\ell-3)/2 \rfloor + 2)\ell t}-1}{q^{\ell t}-1} + 2\left( \frac{h/\ell-3}{2} - \left\lfloor \frac{h/\ell-3}{2} \right\rfloor \right) \ge r\);
    \item \(\alpha \ge 4\) and \(q^{\ell t} + \left\lfloor \frac{h}{\ell} \right\rfloor - \alpha + 1 \ge r\);
    \item \(\alpha \ge 3\) and \(T_1 - T_2 - T_3 > 0\), where \(T_1, T_2, T_3\) are as defined in Lemma~\ref{lem:greedy} with \(n = ht\), \(k = \ell t\), \(m = r\), and \(\alpha\) as given.
\end{itemize}
then a linear network MDS code over \(\mathbb{F}_q^t\) exists for \(\mathcal{N}\).
\end{corollary}

For the generalized combination network \(\GN\) with \(\alpha \le h \le \alpha\ell\), replacing scalar \(\mathbb{F}_q\)-linear MDS codes with their \(\mathbb{F}_q\)-linear counterparts over \(\mathbb{F}_q^t\) gives the following result.

\begin{theorem}\label{thm:thm vector}
For the network \(\mathcal{N}=\GN\) with \(\alpha \le h \le \alpha\ell\),
\[
q_v^{\mathrm{MDS}}(\mathcal{N}) \le \psi(n_{\max} - 1),
\]
where \(n_{\max} = \max(\alpha\ell+\epsilon, r\ell)\). Moreover, if \(h \ge 2\ell\) and a linear network MDS code over \(\mathbb{F}_q^t\) exists for \(\mathcal{N}\), then
\[
r \le \left\lfloor \frac{q^{ht-(\lfloor h/\ell \rfloor-2)\ell t} - 1}{q^{\ell t} - 1} \right\rfloor + \left\lfloor \frac{h}{\ell} \right\rfloor - 2
\quad \text{and} \quad
\alpha\ell + \epsilon - h + 1 \le q^t.
\]
\end{theorem}

\subsection{A Family of Generalized Combination Networks with Zero Gap}

For combination networks \(\mathcal{N}_{h,r,s}\), it is known that vector MDS codes offer no alphabet-size advantage over scalar MDS codes when \(h=2\). Moreover, no analogous results have been established for combination networks with more than two source messages. We have already shown in the previous subsection that for \(\GN\) with \(h=2\), vector MDS codes also offer no advantage. In this subsection, we introduce a family of generalized combination networks with \(2\ell\) source messages, and prove that for these networks, vector MDS codes cannot outperform scalar MDS codes with respect to alphabet size.

Let \(\ell\) and \(r\) be positive integers. We consider a class of generalized combination networks denoted by \((1,\ell)\text{-}\mathcal{N}_{2\ell,r,2\ell+1}\). Each such network contains exactly one direct link from the source node to each sink node. We first recall a fundamental combinatorial object used throughout this section. For an \(n\)-dimensional vector space \(V=\mathbb{F}_{q}^{n}\) over \(\mathbb{F}_q\), a \(t\)-spread \(S\) is a collection of \(t\)-dimensional subspaces of \(V\) such that any two distinct subspaces intersect trivially and their union equals the entire space \(V\). The total number of subspaces in a \(t\)-spread is given by
\[
|S|=\frac{q^{n}-1}{q^t-1}.
\]
A \(t\)-spread is known to exist in \(\mathbb{F}_{q}^{n}\) if and only if \(t\) divides \(n\). See \cite{SE2002} for details.

\begin{construction}\label{cst1}
Let \(S'\) be an \(\ell\)-spread in \(\mathbb{F}_q^{2\ell}\), so that \(|S'|=q^{\ell}+1\). Suppose \(q^{\ell}+1 \ge r\). We arbitrarily select \(r\) subspaces from \(S'\) and label them as \(S=\{s_{1},s_{2},\ldots,s_{r}\}\). For each \(i \in [r]\), let \(B_i \in \mathbb{F}_{q}^{2\ell \times \ell}\) be a matrix whose columns form a basis of \(s_i\).

Consider the generalized combination network \((1,\ell)\text{-}\mathcal{N}_{2\ell,r,2\ell+1}\) with source message vector \((x_{1}, x_{2}, \ldots, x_{2\ell})\). Each middle node indexed by \(i\in [r]\) receives and forwards
\[
y_{i}=(x_{1},x_{2}, \ldots, x_{2\ell})\cdot B_{i} \in \mathbb{F}_{q}^{\ell},
\]
where \(B_i\) acts as the global encoding matrix for all incoming and outgoing edges of the \(i\)-th middle node.

For each sink node \(\gamma=\{i_1,i_2\}\), define the block matrix
\(
B_{\gamma}=\begin{pmatrix}
    B_{i_1} & B_{i_2}
\end{pmatrix}\in \mathbb{F}_{q}^{2\ell \times 2\ell}
\), and partition its columns as \(B_{\gamma}=(b^{1},b^{2},\ldots, b^{2\ell})\). We choose the encoding vector \(a_{\gamma}\) for the direct link associated with sink \(\gamma\) from the set
\begin{equation}
\mathbb{F}_{q}^{2\ell}\setminus \bigcup_{1\leq j_{1}<j_{2}<\cdots<j_{2 \ell-1}\leq 2\ell} \langle 
b^{j_1},b^{j_2},\ldots,b^{j_{2\ell-1}} \rangle,
\label{eq:set}
\end{equation}
where \(\langle\cdot \rangle\) stands for the linear span of the enclosed vectors. Then the sink node \(\gamma=\{i_1,i_2\}\) receives the symbol
\[
(x_{1},x_{2},\ldots, x_{2\ell})\cdot \begin{pmatrix}
  a_{\gamma} & B_{i_1} & B_{i_2}
\end{pmatrix} \in \mathbb{F}_{q}^{2\ell+1}.
\]
\end{construction}

\begin{theorem}\label{thm13}
Construction \ref{cst1} yields a scalar linear MDS code for the generalized combination network \((1,\ell)\text{-}\mathcal{N}_{2\ell,r,2\ell+1}\) with alphabet size \(\psi\left((r-1)^{\frac{1}{\ell}}\right)\).
\end{theorem}

\begin{IEEEproof}
We first verify that the set defined in \eqref{eq:set} is nonempty, which guarantees the validity of Construction \ref{cst1}. Since \(s_{i_1}\) and \(s_{i_2}\) intersect trivially, the columns of \(B_{\gamma}\) form a basis of \(\mathbb{F}_{q}^{2\ell}\). Then a counting argument gives
\[
\left|\mathbb{F}_{q}^{2\ell}\setminus \bigcup_{1\leq j_{1}<j_{2}<\cdots<j_{2 \ell-1}\leq 2\ell} \langle 
b^{j_1},b^{j_2},\cdots,b^{j_{2\ell-1}} \rangle\right|=(q-1)^{2\ell}.
\]
For any prime power \(q\), we have \((q-1)^{2\ell}\ge 1\), so the desired vector \(a_{\gamma}\) always exists. The invertibility of \(B_{\gamma}\) and the choice of \(a_{\gamma}\) ensure that the constructed network code satisfies the MDS property.
\end{IEEEproof}

We now compare the performance of vector and scalar MDS codes for the proposed network by using the lower bound established in previous sections on the minimum alphabet size required for vector MDS codes.

\begin{theorem}\label{thm:thm10}
Let \(\mathcal{N} =(1,\ell)\text{-}\mathcal{N}_{2\ell,r,2\ell+1}\). Then
\[
q_v^{\mathrm{MDS}}(\mathcal{N})=q^{\mathrm{MDS}}(\mathcal{N}).
\]
\end{theorem}

\begin{IEEEproof}
Applying \cref{cor8} with \(h=2\ell\), \(\alpha=2\), and \(\epsilon=1\) gives
\(q_v^{\mathrm{MDS}}(\mathcal{N}) \ge \psi\left((r-1)^{1/\ell}\right)
\). By Theorem \ref{thm13}, the minimum alphabet size for scalar MDS codes satisfies
\[
q^{\mathrm{MDS}}(\mathcal{N}) \le \psi\left((r-1)^{1/\ell}\right).
\]
Combining these two inequalities yields \(q_v^{\mathrm{MDS}}(\mathcal{N})=q^{\mathrm{MDS}}(\mathcal{N})\).
\end{IEEEproof}

This theorem implies that, for the generalized combination network \((1,\ell)\text{-}\mathcal{N}_{2\ell,r,2\ell+1}\), vector MDS codes cannot achieve a smaller alphabet size than scalar MDS codes.

\section{Vector Error-Correcting Codes for Zosin-Khuller Networks}

The Zosin-Khuller network was originally proposed by Zosin and Khuller in \cite{ZK2002} to characterize the integrality gap of linear programming formulations for directed Steiner tree problems. It was later adopted in \cite{CF2006} to illustrate the throughput gain of network coding over traditional routing schemes. In addition, scalar error-correcting codes for this network were studied in \cite{WS2024}. In this section, we extend the investigation to vector error-correcting codes.

We first review fundamental hypergraph concepts required for subsequent analysis. A hypergraph \(H=(V(H),E(H))\) consists of a vertex set \(V(H)\) and an edge set \(E(H)\), where each edge is a nonempty subset of \(V(H)\). A hypergraph is \(r\)-uniform if every edge contains exactly \(r\) vertices. We further review the definition of hypergraph homomorphisms. For two hypergraphs \(G=(V(G),E(G))\) and \(H=(V(H),E(H))\), a mapping \(\phi:V(G)\to V(H)\) is a homomorphism from \(G\) to \(H\) if \(\{\phi(a):a\in e\}\in E(H)\) holds for every edge \(e\in E(G)\). We write \(G\to H\) to denote such a homomorphism. Note that a graph is a $2$-uniform hypergraph. The chromatic number \(\chi(G_1)\) of a graph \(G_1\) is defined as the minimum integer \(n\) such that \(G_1\to K_n\), where \(K_n\) denotes the complete graph on \(n\) vertices. By transitivity of homomorphisms, \(G_1\to G_2\) implies \(\chi(G_1)\le \chi(G_2)\).

Let \(m,N\) be positive integers satisfying \(2 \le m \le N\). Let \(\binom{[N]}{m-1}\) and \(\binom{[N]}{m}\) denote the collections of all \((m-1)\)-element and \(m\)-element subsets of \([N]\), respectively. The Zosin-Khuller network \(Z_{h,m,N}\) is a layered acyclic network with five layers. The source node \(\sigma\) transmits \(h\) messages to \(N\) sink nodes via three intermediate node layers, namely \(A\)-nodes, \(B\)-nodes, and \(C\)-nodes. Specifically, \(A\)-nodes are indexed by the elements of \(\binom{[N]}{m-1}\), while both \(B\)-nodes and \(C\)-nodes are indexed by the elements of \(\binom{[N]}{m}\). The connectivity is as follows: the source node is connected to all \(A\)-nodes; an \(A\)-node is connected to a \(B\)-node if the index subset of the former is contained in that of the latter; each \(B\)-node is connected to the \(C\)-node with the same index; and a sink node \(i \in [N]\) is connected to a \(C\)-node if \(i\) belongs to the index subset of the \(C\)-node. The minimum cut capacity between the source and each sink is \(\bar{C} = \binom{N-1}{m-1}\). An example is illustrated in \cref{fig:ZKmodel}.

\begin{figure}[h]
    \centering
\def\x{0.55}
\begin{tikzpicture}[node distance=2cm, every node/.style={circle, fill, inner sep=2pt}]
    \node (top) at (0, 4) {};
    \node (t1) at (-2.5, 3) {};
    \node (t2) at (-1.5, 3) {};
    \node (t3) at (-0.5, 3) {};
    \node (t4) at (0.5, 3) {};
    \node (t5) at (1.5, 3) {};
    \node (t6) at (2.5, 3) {};
    \node (t11) at (-2, 2) {};
    \node (t12) at (-0.75, 2) {};
    \node (t13) at (0.75, 2) {};
    \node (t14) at (2, 2) {};
    \node (t21) at (-2, 1) {};
    \node (t22) at (-0.75, 1) {};
    \node (t23) at (0.75, 1) {};
    \node (t24) at (2, 1) {};
    \node (b1) at (-3, 0) {};
    \node (b2) at (-1, 0) {};
    \node (b3) at (1, 0) {};
    \node (b4) at (3, 0) {};

    \node[fill=none, inner sep=0pt] at (-2.5, 3.3) {1,2};
    \node[fill=none, inner sep=0pt] at (-1.5, 3.3) {1,3};
    \node[fill=none, inner sep=0pt] at (-0.5, 3.3) {1,4};
    \node[fill=none, inner sep=0pt] at (0.5, 3.3) {2,3};
    \node[fill=none, inner sep=0pt] at (1.5, 3.3) {2,4};
    \node[fill=none, inner sep=0pt] at (2.5, 3.3) {3,4};
    \node[fill=none, inner sep=0pt] at (-2, 2.3) {1,2,3};
    \node[fill=none, inner sep=0pt] at (-0.75, 2.3) {1,2,4};
    \node[fill=none, inner sep=0pt] at (0.75, 2.3) {1,3,4};
    \node[fill=none, inner sep=0pt] at (2, 2.3) {2,3,4};
    \node[fill=none, inner sep=0pt] at (-2, 0.7) {1,2,3};
    \node[fill=none, inner sep=0pt] at (-0.75, 0.7) {1,2,4};
    \node[fill=none, inner sep=0pt] at (0.75, 0.7) {1,3,4};
    \node[fill=none, inner sep=0pt] at (2, 0.7) {2,3,4};
    \node[fill=none, inner sep=0pt] at (-3, -0.3) {1};
    \node[fill=none, inner sep=0pt] at (-1, -0.3) {2};
    \node[fill=none, inner sep=0pt] at (1, -0.3) {3};
    \node[fill=none, inner sep=0pt] at (3, -0.3) {4};
    \draw (top) -- (t1);
    \draw (top) -- (t2);
    \draw (top) -- (t3);
    \draw (top) -- (t4);
    \draw (top) -- (t5);
    \draw (top) -- (t6);
    \draw (t1) -- (t11);
    \draw (t1) -- (t12);
    \draw (t2) -- (t11);
    \draw (t2) -- (t13);
    \draw (t3) -- (t12);
    \draw (t3) -- (t13);
    \draw (t4) -- (t11);
    \draw (t4) -- (t14);
    \draw (t5) -- (t12);
    \draw (t5) -- (t14);
    \draw (t6) -- (t13);
    \draw (t6) -- (t14);

    \draw (t11) -- (t21);
    \draw (t12) -- (t22);
    \draw (t13) -- (t23);
    \draw (t14) -- (t24);

    \draw (t21) -- (b1);
    \draw (t21) -- (b2);
    \draw (t21) -- (b3);
    \draw (t22) -- (b1);
    \draw (t22) -- (b2);
    \draw (t22) -- (b4);
    \draw (t23) -- (b1);
    \draw (t23) -- (b3);
    \draw (t23) -- (b4);
    \draw (t24) -- (b2);
    \draw (t24) -- (b3);
    \draw (t24) -- (b4);
\end{tikzpicture}
\caption{A depiction of the Zosin-Khuller network with $m=3$ and $N=4$.}
   \label{fig:ZKmodel}
 \end{figure}
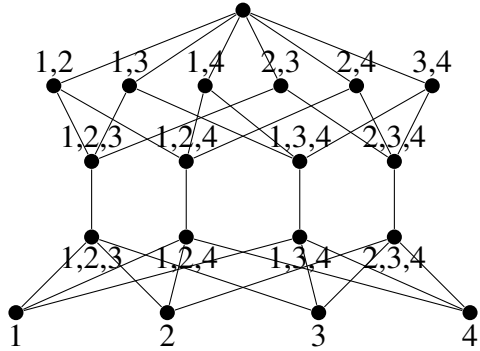

We now introduce two hypergraphs that will be used in our analysis. First, we construct the hypergraph \(H_{N:m}^h\) with vertex set \(\binom{[N]}{m}\). A collection of \(h\) vertices \(\{t_1, \dots, t_h\} \subseteq \binom{[N]}{m}\) forms a hyperedge if and only if the underlying \(m\)-element subsets have nonempty intersection. Second, we recall the definition of \(q\)-analog Kneser hypergraphs. Let \(V \triangleq \mathbb{F}_q^{ht}\). The \(q\)-analog Kneser hypergraph \(qK_{ht:t}^h\) has all \(t\)-dimensional subspaces of \(V\) as its vertices. A set of \(h\) vertices \(\{V_{i_1}, V_{i_2}, \dots, V_{i_h}\}\) constitutes a hyperedge if and only if
\[
\dim(V_{i_1} + \cdots + V_{i_h}) = ht.
\]
For \(h = 2\), we write \(qK_{2t:t}\) for brevity.

\begin{lemma}\label{thm15}
If there exists a vector network MDS code over \(\mathbb{F}_q^t\) for the Zosin-Khuller network \(Z_{h,m,N}\), then there exists a hypergraph homomorphism from \(H_{N:m}^h\) to \(qK_{ht:t}^h\).
\end{lemma}

\begin{IEEEproof}
Let \(\mathcal{C}\) be a vector network MDS code over \(\mathbb{F}_q^t\) for \(Z_{h,m,N}\). For each \(C\)-node indexed by \(j \in \binom{[N]}{m}\), let \(V_j \in \mathbb{F}_q^{ht \times t}\) denote the global encoding vector of \(\mathcal{C}\) associated with its incoming edge. By the MDS property of \(\mathcal{C}\), the decoding matrix at each sink induces an \(\mathbb{F}_q\)-linear \([\bar{C}, h]\) MDS code over \(\mathbb{F}_q^t\), which in particular implies that every \(V_j\) has full column rank. Furthermore, it follows from Lemma \ref{lem: vector linear} that for any \(h\) distinct \(C\)-nodes indexed by \(j_1, \dots, j_h\) whose corresponding subsets contain a common element of \([N]\),
\[
\dim(V_{j_1} + \cdots + V_{j_h}) = ht.
\]
We now define a mapping \(\phi : V(H_{N:m}^h) \to V(qK_{ht:t}^h)\) by setting \(\phi(j) = \langle V_j \rangle\), where \(\langle \cdot \rangle\) denotes the linear span of the matrix columns. By the preceding argument, any hyperedge in \(H_{N:m}^h\) is mapped to a hyperedge in \(qK_{ht:t}^h\). Thus, \(\phi\) is a valid hypergraph homomorphism, which completes the proof.
\end{IEEEproof}

Based on the homomorphism relationship established in \cref{thm15}, we derive a lower bound on the minimum effective field size \(q_v^{\mathrm{MDS}}(Z_{h,m,N})\) for vector MDS codes over Zosin-Khuller networks. To this end, we first summarize known chromatic number results for \(q\)-Kneser graphs.

\begin{lemma}[\cite{BB2012},\cite{CG2006},\cite{I2019}]\label{thm16}
For any prime power \(q\) and positive integer \(t\), the \(q\)-analog Kneser graph \(qK_{2t:t}\) satisfies
\[
\chi(qK_{2t:t}) \le q^t + q^{t-1}.
\]
This bound is tight for \(q \ge 5\) or \(t \le 3\), yielding the exact chromatic number
\[
\chi(qK_{2t:t}) = q^t + q^{t-1}.
\]
\end{lemma}

\begin{theorem}\label{thm17}
Let \(m\), \(N\) be positive integers with \(2 \leq m \leq N\), and let \(\bar{C}=\binom{N-1}{m-1}\geq 2\). The minimum effective field size for vector network MDS codes in \(Z_{2,m,N}\) satisfies
\[q_{v}^{\mathrm{MDS}}(Z_{2,m,N})\geq \tau, \]
where \(\tau=\operatorname{min}\{q^{t}\mid q^{t}+q^{t-1}\geq \binom{N}{m}/\left\lfloor\frac{N}{m}\right\rfloor,\ q~\text{is a prime power},\ t\in \mathbb{N}\}\).
\end{theorem}

\begin{IEEEproof}
Let \(\mathcal{C}\) be a vector network MDS code over \(\mathbb{F}_{q}^{t}\) for \(Z_{2,m,N}\). By \cref{thm15}, there exists a homomorphism \(H_{N:m}^{2}\to qK_{2t:t}\), which implies \(\chi(H_{N:m}^{2})\leq \chi(qK_{2t:t})\).

For the graph \(H_{N:m}^{2}\), any valid vertex coloring partitions its vertex set of \(m\)-subsets into disjoint color classes. Each class contains only pairwise disjoint subsets and thus has size at most \(\lfloor N/m\rfloor\), which yields the lower bound \(\chi(H_{N:m}^{2})\geq \binom{N}{m}/\lfloor N/m\rfloor\).

Combining this result with the chromatic number formula of \(qK_{2t:t}\) in \cref{thm16}, we obtain
\[
\frac{\binom{N}{m}}{\left\lfloor N/m \right\rfloor} \le q^t + q^{t-1}.
\]
By the definition of \(\tau\), the effective field size \(q^t\) of any vector MDS code satisfies \(q^{t}\geq \tau\), completing the proof.
\end{IEEEproof}

\begin{theorem}\label{thm18}
Let \(m,N,h\) be positive integers such that \(3\leq h\leq \bar{C}=\binom{N-1}{m-1}\). The minimum effective field size for vector MDS codes over the Zosin-Khuller network \(Z_{h,m,N}\) satisfies
\[
q_{v}^{\mathrm{MDS}}(Z_{h,m,N})\geq
\begin{cases}
\psi\left(\binom{N}{m}-h+1\right), & m=N-1,\\
\psi\left(\binom{N-1}{m-1}-h+1\right), & 2\leq m\leq N-2.
\end{cases}
\]
\end{theorem}

\begin{IEEEproof}
Let \(\mathcal{C}\) be a vector MDS code over \(\mathbb{F}_q^t\) for the network \(Z_{h,m,N}\). The decoding matrix at each sink induces an \(\mathbb{F}_q\)-linear \([\bar{C},h]\) MDS code over \(\mathbb{F}_q^t\), which corresponds to an \(h\)-\((ht,t,(h-1)t)_q^c\) covering Grassmannian code of size \(\bar{C}\). For \(2 \le m \le N-2\), Lemma \ref{lem:upper bound of covering} gives \(q^t \ge \bar{C} - h + 1\).

For the special case \(m=N-1\), it follows from \cref{thm15} that there exists a hypergraph homomorphism \(\phi: H_{N:m}^h \to qK_{ht:t}^h\). Since \(h\leq N-1\), any \(h\) distinct \((N-1)\)-element subsets of \([N]\) have nonempty intersection, implying that every \(h\)-vertex subset of \(H_{N:m}^h\) forms a hyperedge. Consequently, the entire vertex set of \(H_{N:m}^h\) is mapped via \(\phi\) to an \(h\)-\((ht,t,(h-1)t)_q^c\) covering Grassmannian code of full size \(\binom{N}{m}\). Applying the bound in \cref{lem:upper bound of covering}, we further derive \(q^t \geq \binom{N}{m}-h+1\).
\end{IEEEproof}

\begin{remark}
The scalar network MDS codes over \(Z_{h,m,N}\) have been studied in \cite{WS2024}, where a lower bound on the required field size was derived as \(q^{\mathrm{MDS}}(Z_{h,m,N})\geq \psi\left(\binom{N-1}{m-1}-h+1\right)\).

Consider the two-message scenario, i.e., $h=2$. For $t=1$, the vector-case lower bound established in \cref{thm17} degenerates to a new bound for scalar MDS codes over \(Z_{2,m,N}\), given by
\[
q^{\mathrm{MDS}}(Z_{2,m,N})\geq \psi\left(\frac{\binom{N}{m}}{\left\lfloor N/m \right\rfloor}-1\right).
\]
Whenever $m\nmid N$, we have $\frac{\binom{N}{m}}{\lfloor N/m \rfloor} > \binom{N-1}{m-1}$, so the bound from Theorem~\ref{thm17} is strictly tighter than the prior result in \cite{WS2024}.

Moreover, for $m=N-1$, the inequality $\binom{N}{m}>\binom{N-1}{m-1}$ holds strictly. Hence, the lower bound derived herein again strictly improves the existing bound from \cite{WS2024}.
\end{remark}

\begin{example}
To illustrate the bounds, consider the Zosin-Khuller network with \(N=7\) and \(m=3\). The minimum cut capacity is \(\bar{C} = \binom{7-1}{3-1} = 15\). Combining \cref{thm17} and \cref{thm18}, we derive the following lower bounds on the minimum field size for scalar MDS codes:
\[
q^{\mathrm{MDS}}(Z_{h,3,7}) \ge
\begin{cases}
17, & h=2,\\
\psi(16-h), & 3 \le h \le 15,
\end{cases}
\]
and the following lower bounds for vector MDS codes:
\[
q_v^{\mathrm{MDS}}(Z_{h,3,7}) \ge
\begin{cases}
16, & h=2,\\
\psi(16-h), & 3 \le h \le 15.
\end{cases}
\]

The lower bound for scalar MDS codes given in \cite{WS2024} is \(q^{\mathrm{MDS}}(Z_{h,3,7}) \ge \psi(16-h)\) for all \(2 \le h \le 15\). Our bound improves this to \(17\) for \(h=2\), demonstrating a significant improvement. Note also that for \(h=2\), the lower bound on the minimum field size for vector network MDS codes is strictly smaller than that for scalar MDS codes, suggesting that vector coding may outperform scalar coding in this setting.
\end{example}

\begin{lemma}\label{lem:subset_intersection}
Let \(P_1, P_2, \dots, P_\tau\) be \(m\)-element subsets of \([N]\) such that
\(
\bigcap_{i=1}^{\tau} P_i = \emptyset
\). Then
\[
N \ge \frac{\tau}{\tau-1}\,m.
\]
\end{lemma}

\begin{IEEEproof}
The proof proceeds by induction on \(\tau\). For the case \(\tau=2\), \(P_1 \cap P_2 = \emptyset\) yields \(|P_1 \cup P_2| = 2m\), so \(N \ge 2m\), and the bound holds.

Assume that the claim holds for every collection of at most \(\tau-1\) \(m\)-subsets of \([N]\). Now consider \(\tau\) subsets \(P_1, \dots, P_\tau\) satisfying \(\bigcap_{i=1}^\tau P_i = \emptyset\). Write \(\bigcap_{i=1}^{\tau-1} P_i = \{i_1, i_2, \dots, i_j\}\) for some integer \(0 \le j \le m\), and define \(P_i' = P_i \setminus \{i_1, \dots, i_j\}\) for \(1 \le i \le \tau-1\). By construction, \(\bigcap_{i=1}^{\tau-1} P_i' = \emptyset\). Applying the inductive hypothesis,
\begin{align}\label{eq:lower bound condition}
\left| \bigcup_{i=1}^{\tau-1} P_i' \right| \ge \frac{\tau-1}{\tau-2} (m - j). 
\end{align}
\begin{itemize}
    \item Case 1: \(\left| \bigcup_{i=1}^{\tau-1} P_i' \right| \ge m\). Then
\[
N \ge \left| \bigcup_{i=1}^{\tau-1} P_i \right| = j + \left| \bigcup_{i=1}^{\tau-1} P_i' \right|.
\]
If \(j > m/(\tau-1)\), then \(N \ge j + m > \frac{\tau}{\tau-1}m\). Otherwise, \(j \le m/(\tau-1)\), and
\[
N \ge j + \frac{\tau-1}{\tau-2}(m-j) \ge \frac{\tau}{\tau-1}m.
\]
\item Case 2: \(\left| \bigcup_{i=1}^{\tau-1} P_i' \right| < m\). Since \(\bigcap_{i=1}^\tau P_i = \emptyset\), we have \(\{i_1, \dots, i_j\} \cap P_\tau = \emptyset\), and hence
\[
\left| P_\tau \cap \left( \bigcup_{i=1}^{\tau-1} P_i \right) \right| \le \left| \bigcup_{i=1}^{\tau-1} P_i' \right|.
\]
Thus,
\[
\begin{aligned}
N &\ge \left| \bigcup_{i=1}^{\tau-1} P_i \right| + \left| P_\tau \setminus \bigcup_{i=1}^{\tau-1} P_i \right| \\
&\ge j + \left| \bigcup_{i=1}^{\tau-1} P_i' \right| + m - \left| \bigcup_{i=1}^{\tau-1} P_i' \right| = j + m.
\end{aligned}
\]
From \eqref{eq:lower bound condition} and \(\left| \bigcup_{i=1}^{\tau-1} P_i' \right| < m\), we get \(j > m/(\tau-1)\), so
\[
N \ge j + m > \frac{m}{\tau-1} + m = \frac{\tau}{\tau-1}m.
\]
\end{itemize}
All cases verify \(N \ge \frac{\tau}{\tau-1}m\), completing the induction.
\end{IEEEproof}

\begin{theorem}\label{thm:thm13}
Let \(m, N, h\) be positive integers such that \(m < N\) and
\(
h < \frac{N}{N-m}
\). The minimum effective field size for vector MDS codes over the Zosin-Khuller network \(Z_{h,m,N}\) satisfies
\[
q_v^{\mathrm{MDS}}(Z_{h,m,N}) \ge \psi\left(\binom{N}{m} - h + 1\right).
\]
\end{theorem}

\begin{IEEEproof}
Let \(\mathcal{C}\) be an arbitrary vector MDS code over \(\mathbb{F}_q^t\) for the network \(Z_{h,m,N}\). By Lemma~\ref{thm15}, there exists a hypergraph homomorphism \(\phi: H_{N:m}^h \to qK_{ht:t}^h\). The condition \(h < N/(N-m)\) is equivalent to \(N < \frac{h}{h-1}m\). By Lemma~\ref{lem:subset_intersection}, any \(h\) distinct \(m\)-element subsets of \([N]\) share a common element, implying that every \(h\)-vertex subset of \(H_{N:m}^h\) forms a hyperedge. Consequently, the entire vertex set of \(H_{N:m}^h\) is mapped via \(\phi\) to an \(h\)-\((ht, t, (h-1)t)_q^c\) covering Grassmannian code of full size \(\binom{N}{m}\). Applying the bound from Lemma~\ref{lem:upper bound of covering}, we obtain
\(
q^t \ge \binom{N}{m} - h + 1
\).
\end{IEEEproof}

In what follows, we focus on constructing vector MDS codes for the Zosin--Khuller network \(Z_{h,m,N}\). For \(m > N/2\), the problem of constructing vector network MDS codes reduces to finding a hypergraph homomorphism between two associated hypergraphs.

\begin{theorem}\label{thm:thm14}
For the Zosin--Khuller network \(Z_{h,m,N}\) with \(m > N/2\), a linear network MDS code over \(\mathbb{F}_q^t\) exists if and only if there exists a hypergraph homomorphism
\(
\phi: H_{N:m}^h \to qK_{ht:t}^h
\).
\end{theorem}

\begin{IEEEproof}
The necessity direction follows directly from Lemma~\ref{thm15}. It therefore suffices to establish sufficiency.

When \(m > N/2\), the number of \(A\)-nodes is at least the number of \(B\)-nodes. Consider the bipartite subgraph of \(Z_{h,m,N}\) induced by its \(A\)-node and \(B\)-node partitions. An edge connects an \(A\)-node to a \(B\)-node precisely when the index set of the \(A\)-node is contained within that of the \(B\)-node. Under this adjacency rule, every \(A\)-node has degree \(N-m+1\), and every \(B\)-node has degree \(m\). Since \(N-m+1 \le m\), Hall's marriage theorem guarantees a matching covering all \(B\)-nodes.

Delete all unmatched \(A\)-nodes and their incident edges to obtain a subnetwork \(S\) of \(Z_{h,m,N}\). In \(S\), the counts of \(A\)-nodes, \(B\)-nodes, and \(C\)-nodes are equal, each taking value \(\binom{N}{m}\). For each \(C\)-node, retain \(\binom{N}{m}\) edge-disjoint source-to-\(C\)-node paths as well as edges connecting \(C\)-nodes to sinks, yielding a subnetwork \(S'\). For later construction, we assign identical index labels to the \(A\)-nodes, \(B\)-nodes, and their corresponding \(C\)-nodes along each source-to-\(C\)-node path.

A linear MDS code defined over \(S'\) automatically induces a valid linear MDS code for the full network \(Z_{h,m,N}\). Suppose such a hypergraph homomorphism exists, denoted \(\phi: V(H_{N:m}^h) \to V(qK_{ht:t}^h)\). We construct a vector coding scheme for \(Z_{h,m,N}\) using only edges within \(S'\), and local encoding vectors for edges outside \(S'\) are set to the zero matrix. For each \(A\)-node indexed by \(k \in \binom{[N]}{m}\), define the local encoding vector \(V_k\) on its incoming link using a basis of the subspace \(\phi(k)\). All \(A\)-nodes, \(B\)-nodes, and \(C\)-nodes simply forward received packets without additional processing.

We now verify that the constructed code satisfies the MDS property. Fix an arbitrary sink \(i \in [N]\), and let \(\{i_1, i_2, \dots, i_{\bar{C}}\}\subseteq \binom{[N]}{m}\) be the collection of \(C\)-nodes connected to \(i\), where \(\bar{C} = \binom{N-1}{m-1}\). The extended global encoding matrix at sink \(\gamma\) has the form
\[
\begin{pmatrix} F(\gamma) \\ G(\gamma) \end{pmatrix},
\]
where \(F(\gamma) = \begin{pmatrix} V_{i_1} & V_{i_2} & \cdots & V_{i_{\bar{C}}} \end{pmatrix}\). Viewing \(G(\gamma)\) as a block matrix composed of \(t \times t\) subblocks, each row of \(G(\gamma)\) contains at most one nonzero block. Take any \(h\)-element subset \(\{i_{j_1}, \dots, i_{j_h}\} \subseteq \{i_1, i_2, \dots, i_{\bar{C}}\}\). By the homomorphism property, we get that the block matrix
\[
\begin{pmatrix} V_{i_{j_1}} & V_{i_{j_2}} & \cdots & V_{i_{j_h}} \end{pmatrix}
\]
has full rank. This confirms that the constructed scheme is a valid linear MDS code over \(\mathbb{F}_q^t\) for \(Z_{h,m,N}\).
\end{IEEEproof}

\begin{corollary}
Let \(m, N, h\) be positive integers such that \(N/2 < m \le N\) and \(h \le \bar{C} = \binom{N-1}{m-1}\). The minimum effective field size for vector MDS codes in \(Z_{h,m,N}\) satisfies
\[
q_v^{\mathrm{MDS}}(Z_{h,m,N}) \le \eta,
\]
where
\(
\eta = \min \left\{ q^t \mid \exists \, H_{N:m}^h \to qK_{ht:t}^h \text{ for prime power } q \text{ and positive integer } t \right\}
\).
\end{corollary}

\begin{remark}
A trivial upper bound for \(\eta\) is given by \(\psi(\binom{N}{m}-1)\). Indeed, Lemmas \ref{lem:2lower bound of covering}, \ref{lem:3lower bound of covering}, and \ref{lem:4lower bound} guarantee the existence of an \(h\)-\((ht, t, (h-1)t)_q^c\) covering Grassmannian code with \(\binom{N}{m}\) codewords whenever \(q^t \ge \binom{N}{m} - 1\). This yields
\[
q_v^{\mathrm{MDS}}(Z_{h,m,N}) \le \psi\left(\binom{N}{m} - 1\right).
\]
Tighter upper bounds for \(q_v^{\mathrm{MDS}}(Z_{h,m,N})\) can be obtained by analyzing the underlying combinatorial structure of the hypergraph \(H_{N:m}^h\).
\end{remark}

We now establish upper bounds on the MDS gap for the Zosin--Khuller networks.

\begin{corollary}\label{cor:gap}
Let \(m, N, h\) be positive integers such that \(h \le \binom{N-1}{m-1}\). Then
\[
\operatorname{gap}^{\mathrm{MDS}}(Z_{h,2,N})
\le
\begin{cases}
\psi_{\mathrm{even}}(N) - \tau, & h = 2, \\[6pt]
\psi_{\mathrm{even}}(N) - \psi(N - h ), & \text{otherwise},
\end{cases}
\]
and when \(m > N/2\),
\[
\operatorname{gap}^{\mathrm{MDS}}(Z_{h,m,N})
\le
\begin{cases}
\psi\left(\binom{N}{m}-1\right) - \tau, & h = 2, \\[6pt]
\psi\left(\binom{N}{m}-1\right) - \psi\left(\binom{N}{m}-h+1\right), & h < \frac{N}{N-m}, \\[6pt]
\psi\left(\binom{N}{m}-1\right) - \psi\left(\binom{N-1}{m-1} - h + 1\right), & \text{otherwise},
\end{cases}
\]
where \(\psi_{\mathrm{even}}(N)\) denotes the smallest even prime power greater than or equal to \(N\), and
\(
\tau = \min\left\{ q^t \;\middle|\; q^t + q^{t-1} \ge \binom{N}{m}/{\left\lfloor \frac{N}{m} \right\rfloor},\ q \text{ is a prime power},\ t \in \mathbb{N} \right\}.
\)
\end{corollary}

\section{Summary and Concluding Remarks}

In this paper, we systematically investigated the minimum field size required for scalar and vector network MDS codes over two families of network topologies: generalized combination networks $(\epsilon,\ell)$-$\mathcal{N}_{h,r,\alpha\ell+\epsilon}$ and Zosin–Khuller networks $Z_{h,m,N}$.

For generalized combination networks, we established rigorous equivalences linking the minimum distance of scalar network codes to the minimum Hamming distance of classical linear codes, and that of vector network codes to the minimum Hamming distance of $\mathbb{F}_{q}$-linear vector codes. These equivalences convert network-level MDS constraints into conventional coding-theoretic conditions. Depending on the parameter regimes, namely $h\leq \alpha$, $\alpha\leq h\leq \alpha\ell$, and $\alpha\ell\leq h\leq \alpha\ell+\epsilon$, we derived necessary and sufficient conditions for the existence of scalar and vector network MDS codes, connecting them to classical MDS codes and covering Grassmannian codes. Using refined greedy constructions and MRD-based designs, we obtained improved lower and upper bounds on the minimum field size. Notably, our upper bound significantly outperforms the state-of-the-art universal bound for general network MDS codes in \cite{GY2021}. Furthermore, we identified a new family of generalized combination networks $(1,\ell)$-$\mathcal{N}_{2\ell,r,2\ell+1}$ with zero MDS gap, demonstrating that vector coding offers no field-size advantage over scalar coding in these configurations.

For Zosin–Khuller networks, we developed a hypergraph homomorphism framework for vector network MDS codes and established a necessary and sufficient condition for their existence. For $h=2$, using chromatic number arguments, we derived a lower bound for vector MDS codes, which in the scalar case yields \(q^{\mathrm{MDS}}(Z_{2,m,N})\geq \psi\left(\binom{N}{m}/\left\lfloor\frac{N}{m}\right\rfloor-1\right)\), strictly improving the prior bound from \cite{WS2024}. For parameters satisfying $h<N/(N-m)$, applying subset-intersection arguments, we obtained a lower bound $q_{v}^{\mathrm{MDS}}(Z_{h,m,N})\geq \psi\left(\binom{N}{m}-h+1\right)$, which also improves the corresponding bound in \cite{WS2024}. For general $h$, we provided lower and upper bounds on the minimum field size for vector MDS codes, as well as upper bounds on the MDS gap between optimal scalar and vector solutions.

Several open problems remain for future investigation. For generalized combination networks with $\epsilon\neq 0$, it remains unknown whether vector MDS codes can achieve a strictly smaller field size than scalar MDS codes, since the direct links between the source and the sink nodes are crucial for vector codes to have an advantage in such networks. For Zosin–Khuller networks, tightening the gap between our derived lower and upper bounds on the minimum field size is a natural next step. Finally, extending the hypergraph homomorphism framework to other network topologies beyond ZK networks may further clarify the inherent limitations and potential advantages of vector network coding.


\section*{Appendix A}

\subsection{Proof of \cref{lem:3lower bound of covering}}

We take the MRD code \(\mathscr{C}\) constructed in Section~\ref{sec:covering_bounds}  as the fundamental building block to construct the target covering Grassmannian codes and further derive the corresponding lower bound.

Let \(n=3k+2mk\) for some non-negative integer \(m\). Define
\[B_0 = \left\{ \begin{pmatrix} I_k \\ c \\ c^2 \end{pmatrix} \,\bigg|\, c\in \mathscr{C} \right\} \cup \left\{ \begin{pmatrix} 0_k \\ 0_k \\ I_k \end{pmatrix} \right\}.\]
For each integer \(1\le j\le m\), recursively set
\[B_j = \left\{ \begin{pmatrix} b \\ c \\ c^2 \end{pmatrix} \,\bigg|\, b\in B_{j-1},\,c\in \mathscr{C} \right\} \cup \left\{ \begin{pmatrix} 0_k \\ \vdots \\ 0_k \\ I_k \end{pmatrix} \right\}.\]
Let \(S(B_j)=\{\langle b\rangle \mid b\in B_j\}\), where \(\langle b\rangle\) denotes the subspace spanned by the columns of the matrix \(b\). We verify that \(S(B_m)\) constitutes a \(3\text{-}(n,k,2k)_q^c\) covering Grassmannian code by induction on \(m\).

We first consider \(m=0\). For any distinct \(c_1,c_2,c_3\in \mathscr{C}\), we compute
\[\det\begin{pmatrix} I_k & I_k & I_k \\ c_1 & c_2 & c_3 \\ c_1^2 & c_2^2 & c_3^2 \end{pmatrix} =\det\begin{pmatrix} I_k & 0 & 0 \\ c_1 & c_2-c_1 & c_3-c_1 \\ c_1^2 & c_2^2-c_1^2 & c_3^2-c_1^2 \end{pmatrix}= \det\begin{pmatrix} c_2-c_1 &c_3-c_1 \\  c_2^2-c_1^2 & c_3^2-c_1^2 \end{pmatrix}.\]
Since \(\mathscr{C}\) is an MRD code, \(c_2-c_1\), \(c_3-c_1\) and \(c_3-c_2\) are invertible matrices. It follows that
\begin{equation}
\begin{split}
\det\begin{pmatrix} c_2-c_1 & c_3-c_1 \\ c_2^2-c_1^2 & c_3^2-c_1^2 \end{pmatrix} &=\det\left( \begin{pmatrix} I_k & I_k \\ c_1+c_2 & c_1+c_3 \end{pmatrix} \begin{pmatrix} c_2-c_1 & 0 \\ 0 & c_3-c_1 \end{pmatrix} \right)\\
&=\det \begin{pmatrix} I_k & 0 \\ c_1+c_2 & c_3-c_2\end{pmatrix}\cdot\det\begin{pmatrix} c_2-c_1 & 0 \\ 0 & c_3-c_1 \end{pmatrix}  \\
 &=\det (c_3-c_2)\cdot\det(c_2-c_1)\cdot\det(c_3-c_1)\neq 0.\end{split}\label{eq:det}\end{equation}
Furthermore, for arbitrary distinct \(c_1,c_2\in \mathscr{C}\),
\[\det\begin{pmatrix} I_k & I_k & 0_k \\ c_1 & c_2 & 0_k \\ c_1^2 & c_2^2 & I_k \end{pmatrix} =\det\begin{pmatrix} I_k & I_k \\ c_1 & c_2 \end{pmatrix} =\det(c_2-c_1)\neq 0.\]
Hence \(S(B_0)\) is a \(3\text{-}(3k,k,2k)_q^c\) covering Grassmannian code.

Suppose the statement holds for all positive integers \(j < m\), i.e., assume that \(S(B_{j})\) is a \(3\text{-}(3k+2jk,k,2k)_q^c\) covering Grassmannian code for all \(j<m\). We now prove the assertion for \(j=m\).
Take three arbitrary elements \(b_{1}\), $b_{2}$ and $b_{3}$ from \(B_m\) and analyze their linear span in two cases.

\begin{itemize}
    \item Assume that the matrix \(\begin{pmatrix}0_k\\\vdots\\0_k\\I_k\end{pmatrix}\) is contained in the three selected elements. Without loss of generality, let
    \(b_{1}=\begin{pmatrix}0_k\\\vdots\\0_k\\I_k\end{pmatrix},\; b_{2}=\begin{pmatrix}b'_{2}\\c_{2}\\c_{2}^2\end{pmatrix},\; b_{3}=\begin{pmatrix}b'_{3}\\c_{3}\\c_{3}^2\end{pmatrix},\)
    where \(b'_{2},b'_{3}\in B_{m-1}\). We have
    \[\mathrm{rank}\begin{pmatrix} b'_{2} & b'_{3} & 0\\ c_{2} & c_{3} & 0_k\\ c_{2}^2 & c_{3}^2 & I_k \end{pmatrix} =\mathrm{rank}\begin{pmatrix}b'_{2} & b'_{3}\\c_{2} & c_{3}\end{pmatrix}+k.\]
    \begin{itemize}
        \item If \(b'_{2}\neq b'_{3}\), the induction hypothesis yields \(\mathrm{rank}\begin{pmatrix}b'_{2} & b'_{3}\end{pmatrix}=2k\), so the total rank equals \(3k\).
        \item If \(b'_{2}=b'_{3}\), then \(c_{2}\neq c_{3}\), and
    \(\mathrm{rank}\begin{pmatrix}b'_{2} & b'_{2}\\c_{2} & c_{3}\end{pmatrix} =\mathrm{rank}\begin{pmatrix}b'_{2} & 0\\c_{2} & c_{3}-c_{2}\end{pmatrix}=2k.\)
    The overall rank is still \(3k\).
    \end{itemize}
    \item Assume that the matrix \(\begin{pmatrix}0_k\\\vdots\\0_k\\I_k\end{pmatrix}\) is not contained in the three selected elements. Write
\(b_{1}=\begin{pmatrix}b'_{1}\\c_{1}\\c_{1}^2\end{pmatrix},\; b_{2}=\begin{pmatrix}b'_{2}\\c_{2}\\c_{2}^2\end{pmatrix},\; b_{3}=\begin{pmatrix}b'_{3}\\c_{3}\\c_{3}^2\end{pmatrix},\)
with \(b'_{1},b'_{2},b'_{3}\in B_{m-1}\).
\begin{itemize}
    \item If \(b'_{1},b'_{2},b'_{3}\) are mutually distinct, then \(\mathrm{rank}\begin{pmatrix}b'_{1} & b'_{2} & b'_{3}\end{pmatrix}=3k\) by induction.
    \item If \(b'_{1}=b'_{2}=b'_{3}\), then \(c_{1},c_{2},c_{3}\) are distinct. And
    \begin{align*}
    \mathrm{rank}\begin{pmatrix} b'_{1} & b'_{1} & b'_{1} \\ c_1 & c_2 & c_3 \\ c_1^2 & c_2^2 & c_3^2 \end{pmatrix}
    &=\mathrm{rank}\begin{pmatrix} b'_{1} & 0 & 0 \\ c_1 & c_2-c_1 & c_3-c_1 \\ c_1^2 & c_2^2-c_1^2 & c_3^2-c_1^2 \end{pmatrix}\\
    &=\mathrm{rank}\begin{pmatrix}  c_2-c_1 & c_3-c_1 \\ c_2^2-c_1^2 & c_3^2-c_1^2 \end{pmatrix}+k.
    \end{align*}
    From \eqref{eq:det}, the last matrix has rank \(2k\). Consequently, the overall rank is \(3k\).
    \item If exactly two entries coincide, similar rank computation also leads to full rank \(3k\).
    \end{itemize}
\end{itemize}
All cases verify \(\mathrm{rank}(b_{1},b_{2},b_{3})=3k\). Hence \(S(B_m)\) is a \(3\text{-}(3k+2mk,k,2k)_q^c\) covering Grassmannian code. The cardinality satisfies the recurrence \(|B_j|=|B_{j-1}|\cdot|\mathscr{C}|+1\) with initial value \(|B_0|=q^k+1\). Solving the recurrence yields
\(|B_m|=\frac{q^{(m+2)k}-1}{q^k-1}.\)

Furthermore, for the case \(n=3k+(2m+1)k\) where \(m\) is a non-negative integer, define
\[\hat{B}_m=\left\{\begin{pmatrix}b\\0_k\end{pmatrix}\,\bigg|\,b\in B_m\right\} \cup\left\{\begin{pmatrix}0_k\\\vdots\\0_k\\I_k\end{pmatrix}\right\}.\]
Then \(S(\hat{B}_m)\) forms a \(3\text{-}(3k+(2m+1)k,k,2k)_q^c\) covering Grassmannian code of size
\(|\hat{B}_m|=|B_m|+1=\frac{q^{(m+2)k}-1}{q^k-1}+1.\)

\subsection{Proof of \cref{lem:4lower bound}}

Let $r$ be the remainder when $n$ is divided by $k$, and define \(k^{\prime} = \lfloor \frac{n}{k}\rfloor\). We construct the target covering Grassmannian code based on the MRD code \(\mathscr{C}\) constructed in \cref{sec:covering_bounds}. For each matrix \(B\in \mathscr{C}\), define a \(k\times n\) block matrix \(M_B\) consisting of \(k'\) full \(k\times k\) square blocks and one residual \(k\times r\) zero block:
\[
M_{B}=\begin{pmatrix}
I_{k}&B&B^{2}&\cdots &B^{k'-1} & \mathbf{0}_{k\times r}
\end{pmatrix},
\]
where \(B^i\) denotes the \(i\)-th power of \(B\). Let
\[
\mathcal{M} = \big\{ \langle M_{B} \rangle \,\big|\, B\in \mathscr{C} \big\},
\]
where \(\langle M_B\rangle\) denotes the row space of \(M_B\) over \(\mathbb{F}_q\).

For each integer \(1\le i\le k'\), define the block unit row matrix
\[
\mathbf{E}_i = \begin{pmatrix}
\mathbf{0}_{k\times k} & \cdots & I_k & \cdots & \mathbf{0}_{k\times k} & \mathbf{0}_{k\times r}
\end{pmatrix}\in \mathbb{F}_q^{k\times n},
\]
where the \(i\)-th \(k\times k\) block is the identity matrix \(I_k\) and all other blocks are zero matrices. Since \(n\ge \alpha k\), we have \(k'\ge \alpha\). Define
\[
\mathcal{S} = \big\{ \langle \mathbf{E}_i \rangle \,\big|\, \alpha \le i \le k' \big\},
\]
which implies \(|\mathcal{S}| = k'-\alpha+1\).

We now verify that \(\mathcal{D} = \mathcal{M}\cup \mathcal{S}\) forms an \(\alpha\)-\((n,k,(\alpha-1)k)_q^c\) covering Grassmannian code. By definition, it suffices to prove that any \(\alpha\) distinct subspaces in \(\mathcal{D}\) span an \(\alpha k\)-dimensional subspace of \(\mathbb{F}_q^n\).

Take an arbitrary \(\alpha\)-subset \(\{t_1,t_2,\dots,t_\alpha\}\subseteq \mathcal{D}\). We first consider the case where all subspaces are chosen from \(\mathcal{M}\). In this scenario, \(t_i = \langle M_{B_i}\rangle\) for pairwise distinct \(B_i\in \mathscr{C}\). Consider the Vandermonde-type block matrix
\[
F=\begin{pmatrix}
I_{k} & B_{1} & B_{1}^{2} & \dots & B_{1}^{\alpha-1} \\
I_{k} & B_{2} & B_{2}^{2} & \dots & B_{2}^{\alpha-1} \\
\vdots & \vdots & \vdots & \ddots & \vdots \\
I_{k} & B_{\alpha} & B_{\alpha}^{2} & \dots & B_{\alpha}^{\alpha-1}
\end{pmatrix}.
\]
All matrices in the MRD code \(\mathscr{C}\) commute with each other under matrix multiplication. Meanwhile, the difference of any two distinct matrices \(B_i,B_j\) is full rank. This guarantees that \(F\) admits a determinant factorization identical to that of the classical Vandermonde matrix, which further implies that \(F\) is nonsingular, i.e., \(\det(F)\neq 0\). Consequently,
\[
\dim\big(t_1+t_2+\cdots+t_\alpha\big)=\alpha k.
\]

We next consider the case where \(m\) subspaces are selected from \(\mathcal{S}\) and the remaining \(\alpha-m\) subspaces are selected from \(\mathcal{M}\), with \(1\le m\le \min(\alpha,k^{\prime}-\alpha+1)\). The subspaces in \(\mathcal{S}\) correspond to distinct block positions, whose row spaces are blockwise linearly independent. Combined with the full-rank property of the Vandermonde block structure of matrices in \(\mathscr{C}\), the overall sum of these \(\alpha\) subspaces still has dimension \(\alpha k\).

Therefore, \(\mathcal{D}\) is a valid \(\alpha\)-\((n,k,(\alpha-1)k)_q^c\) covering Grassmannian code. Its cardinality is
\(
|\mathcal{D}|= |\mathcal{M}| + |\mathcal{S}| = q^k + k' - \alpha + 1
\). Substituting \(k' = \lfloor \frac{n}{k} \rfloor\), we obtain the following lower bound:
\[
B_q\big(n, k, (\alpha-1)k; \alpha\big) \ge q^k + \left\lfloor \frac{n}{k} \right\rfloor - \alpha + 1.
\]

\section*{Appendix B}

\section*{Proof of \cref{lem:greedy}}

We construct the desired \(\alpha\)-\((n,k,(\alpha-1) k)_q^c\) covering Grassmannian code via a greedy sequential method.

Let \(\mathcal{D}\) be an empty collection. We first select \(\alpha-1\) mutually linearly independent \(k\)-dimensional subspaces \(V_1,V_2,\dots,V_{\alpha-1}\in\mathcal{G}_q(n,k)\). We then construct the \(\alpha\)-th subspace \(V_\alpha\) by sequentially choosing its basis vectors. Define the initial valid vector set
\[
S_{\alpha,1} = \mathbb{F}_q^n \setminus  \sum_{i=1}^{\alpha-1} V_i,
\]
and select an arbitrary vector \(v_{\alpha,1}\in S_{\alpha,1}\). For \(2\le j\le k\), recursively set
\[
S_{\alpha,j} = \mathbb{F}_q^n \setminus \left( \sum_{i=1}^{\alpha-1} V_i + \langle v_{\alpha,1},\dots,v_{\alpha,j-1}\rangle \right),
\]
and pick \(v_{\alpha,j}\in S_{\alpha,j}\). The condition \(n\ge\alpha k\) guarantees \(S_{\alpha,j}\neq\emptyset\) for all \(j\), ensuring the successful construction of the \(k\)-dimensional subspace \(V_\alpha\). By construction,
\[
\dim\left(\sum_{i=1}^{\alpha}V_i\right)=\alpha k,
\]
which satisfies the property of the target code. We now obtain the initial valid code set \(\mathcal{D}=\{V_1,V_2,\dots,V_\alpha\}\).

For each integer \(m\ge \alpha+1\), we construct the subsequent subspace \(V_{m}\) following a unified validity rule. Define the candidate set for the first basis vector of \(V_{m}\) as
\[
S_{m,1} = S_{m-1,1} \setminus \bigcup_{\substack{A \subseteq [m-2]: \\ |A|=\alpha-2}} \left( \sum_{i \in A} V_i + V_{m-1} \right),
\]
and select \(v_{m,1}\in S_{m,1}\). For \(2\le j\le k\), further define
\[
S_{m,j} = S_{m,1} \setminus \bigcup_{\substack{A \subseteq [m-1]: \\ |A|=\alpha-1}} \left( \sum_{i \in A} V_i + \langle v_{m,1}, \dots, v_{m,j-1}\rangle \right),
\]
and select a valid basis vector \(v_{m,j}\in S_{m,j}\). The extension terminates when \(S_{\tau,k}\) becomes empty for some index \(\tau\). The construction ensures that any \(\alpha\) distinct subspaces in the resulting collection span an \(\alpha k\)-dimensional space.

We next calculate the maximum number of constructible codewords. For the initial construction, the valid vector set satisfies
\[
|S_{\alpha,1}|\geq q^n-q^{(\alpha-1)k}.
\]
For each extension step, the candidate set obeys
\[
|S_{m,1}|\geq |S_{m-1,1}|-(q^k-1)-\binom{m-2}{\alpha-2}(q^k-1)\big(q^{(\alpha-2)k}-1\big),
\]
and
\[
|S_{m,k}|\geq |S_{m,1}|-(q^{k-1}-1)-\binom{m-1}{\alpha-1}(q^{k-1}-1)\big(q^{(\alpha-1)k}-1\big).
\]
Iterative substitution yields the unified lower bound
\[
\begin{aligned}
|S_{m,k}|&\geq q^n-q^{(\alpha-1)k}-(m-\alpha)(q^k-1)-\sum_{i=\alpha-1}^{m-2}\binom{i}{\alpha-2}(q^k-1)\big(q^{(\alpha-2)k}-1\big) \\
&\quad -(q^{k-1}-1)-\binom{m-1}{\alpha-1}(q^{k-1}-1)\big(q^{(\alpha-1)k}-1\big).
\end{aligned}
\]

Therefore, for any positive integer \(m\) such that the right-hand side of the above inequality is positive, a valid \(\alpha\)-\((n,k,(\alpha-1) k)_q^c\) covering Grassmannian code with exactly \(m\) codewords can be constructed by the greedy procedure.

\bibliographystyle{IEEEtran}  
\bibliography{references}



\end{document}